%% file: Async_pII_arxiv_v5.tex
\documentclass[11pt,journal,onecolumn,draftcls]{IEEEtran}

\IEEEoverridecommandlockouts

\input MyDefinitions.tex

\begin{document}
\title{Asynchronous Adaptation and Learning over Networks --- Part II: Performance Analysis}

\author{Xiaochuan~Zhao,~\IEEEmembership{Student~Member,~IEEE,}
        and Ali~H.~Sayed,~\IEEEmembership{Fellow,~IEEE} \\
\thanks{The authors are with Department of Electrical Engineering, University of California, Los Angeles, CA 90095
Email: \{xzhao,~sayed\}@ee.ucla.edu.}% <-this % stops a space
\thanks{This work was supported by NSF grants CCF-1011918 and ECCS-1407712. A short and limited early version of this work appeared in the conference proceeding \cite{Zhao12EUSIPCO}. The first part of this work is presented in \cite{Zhao13TSPasync1}.}}

%\markboth{IEEE~TRANSACTIONS~ON~SIGNAL~PROCESSING,~VOL.~xx,
%~NO.~X,~MONTH~YEAR}
%{Zhao and Sayed: Asynchronous Adaptation and Learning over Networks --- Part II: Performance Analysis}

\maketitle
% ------

\begin{abstract}
In Part I \cite{Zhao13TSPasync1}, we introduced a fairly general model for asynchronous events over adaptive networks including random topologies, random link failures, random data arrival times, and agents turning on and off randomly. We performed a stability analysis and established the notable fact that the network is still able to converge in the mean-square-error sense to the desired solution. Once stable behavior is guaranteed, it becomes important to evaluate how fast the iterates converge and how close they get to the optimal solution. This is a demanding task due to the various asynchronous events and due to the fact that agents influence each other. In this Part II, we carry out a detailed analysis of the mean-square-error performance of asynchronous strategies for solving distributed optimization and adaptation problems over networks. We derive analytical expressions for the mean-square convergence rate and the steady-state mean-square-deviation. The expressions reveal how the various parameters of the asynchronous behavior influence network performance. In the process, we establish the interesting conclusion that even under the influence of asynchronous events, all agents in the adaptive network can still reach an $O(\nu^{1 + \gamma_o'})$ near-agreement with some $\gamma_o' > 0$ while approaching the desired solution within $O(\nu)$ accuracy, where $\nu$ is proportional to the small step-size parameter for adaptation.
\end{abstract}


\allowdisplaybreaks

\section{Introduction}
In Part I \cite{Zhao13TSPasync1}, we introduced a fairly general model for asynchronous distributed adaptation and learning over networks. The model allows the step-size for any agent to be randomly chosen within a range of values including zero, meaning an off-status; it also allows the communication links between any two agents to be randomly turned on and off; it further allows the topology to vary and the combination weights on the links between agents to vary randomly. The model also captures the situation in which agents randomly select a subset of their neighbors to share information with,  as happens in gossip implementations. Based on this asynchronous model, we carried out a detailed stability analysis in Part I \cite{Zhao13TSPasync1} and arrived at Theorem \ref{I-theorem:meansquarestability} in that Part, which provides an explicit condition on the first and second-order moments of the distribution of the step-sizes to ensure mean-square stability of the adaptive network. When the condition holds, it was further shown that the asymptotic mean-square-deviation (MSD) for every agent in the network remains bounded. Interestingly, it was shown that these conclusions hold irrespective of the randomness in the network topology.

In this Part II, we conduct a detailed mean-square-error (MSE) analysis in order to characterize the learning behavior of the asynchronous network in terms of the network parameters. In particular, we answer the following three questions:
\begin{itemize}
\item How is the convergence rate of the algorithm affected by the occurrence of asynchronous events in comparison to a synchronous network?

\item Are agents still able to reach some sort of agreement in steady-state despite the various sources of randomness influencing their interactions?

\item How close do the iterates of the various agents get to each other and to the desired optimal solution?
\end{itemize}
One of the main conclusions that will follow from the analysis is that, under certain reasonable conditions, the asynchronous network will continue to be able to deliver performance that is comparable to the synchronous case where no failures occur. In particular, we will be able to establish that for a sufficiently small step-size parameter $\nu$ it holds that
\begin{align}
\label{eqn:individualMSDapproxOnu}
\lim_{i\rightarrow \infty} \E\,\|w^o - \bm{w}_{k,i}\|^2 & = O(\nu) \\
\label{eqn:crossMSDapproxOnu2}
\lim_{i\rightarrow \infty} \E\,\|\bm{w}_{k,i} - \bm{w}_{\ell,i}\|^2 & = O(\nu^{1 + \gamma_o'})
\end{align}
for some $\gamma_o' > 0$ that is given by \eqref{eqn:gammaopdef} further ahead. These results imply that all agents reach a level of $O(\nu^{1 + \gamma_o'})$ agreement with each other and get $O(\nu)$ close to the desired optimal solution, $w^o$, in steady-state. This interesting behavior is illustrated in Fig. \ref{fig:normball}, where it is shown that, despite being subjected to various sources of randomness, the agents in an asynchronous network are still able to approach the desired solution and they are also able to coalesce close to each other. For the remainder of this part, we continue to use the same symbols, notation, and modeling assumptions from Part I \cite{Zhao13TSPasync1}. Moreover, we focus on presenting the main results and their interpretation in the body of the paper, while delaying the detailed proofs and arguments to the appendices.

\emph{Notation}: We use lowercase letters to denote vectors, uppercase letters for matrices, plain letters for deterministic variables, and boldface letters for random variables. We also use $(\cdot)^{\T}$ to denote transposition, $(\cdot)^{*}$ to denote conjugate transposition, $(\cdot)^{-1}$ for matrix inversion, $\Tr(\cdot)$ for the trace of a matrix, $\lambda(\cdot)$ for the eigenvalues of a matrix, $\| \cdot \|$ for the 2-norm of a matrix or the Euclidean norm of a vector, and $\rho(\cdot)$ for the spectral radius of a matrix. Besides, we use $\kron$ to denote the Kronecker product and $\bkron$ to denote the block Kronecker product.

\begin{figure}
\includegraphics[scale=1]{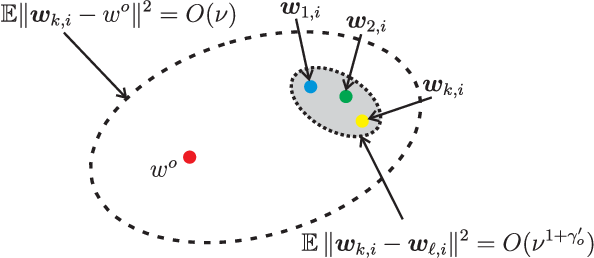}
\centering
\caption{Illustration for Eqs. \eqref{eqn:individualMSDapproxOnu} and \eqref{eqn:crossMSDapproxOnu2}: the agents do not only get $O(\nu)$ close to the target $w^o$ but they also cluster next to each other within $O(\nu^{1 + \gamma_o'})$ for some $\gamma_o' > 0$.}
\label{fig:normball}
\vspace{-1\baselineskip}
\end{figure}

\section{Network Error Dynamics}

In Part I \cite{Zhao13TSPasync1}, we examined the mean-square stability of the following asynchronous diffusion strategy:
\begin{subequations}
\begin{align}
\label{eqn:atcasync1}
{\bm{\psi}}_{k,i} & = {\bm{w}}_{k,i-1} - \bm{\mu}_k(i) \wh{\nabla_{w^*}J_k}(\bm{w}_{k,i-1}) \\
\label{eqn:atcasync2}
{\bm{w}}_{k,i} & = \sum_{\ell\in\bm{\Ncal}_{k,i}} \bm{a}_{\ell k}(i)\,{\bm{\psi}}_{\ell,i}
\end{align}
\end{subequations}
for the distributed solution of a global inference task in the following form:
\be
\label{eqn:globalcost}
\minimize_{w}\;\;J^{\glob}(w)\defeq\sum_{k=1}^{N}J_k(w)
\ee
The global optimum of \eqref{eqn:globalcost} is denoted by $w^o$. We continue to assume the same asynchronous model from Part I \cite{Zhao13TSPasync1}. Specifically, the random step-sizes $\{\bm{\mu}_k(i)\}$ and random combination coefficients $\{\bm{a}_{\ell k}(i)\}$ in \eqref{eqn:atcasync1}--\eqref{eqn:atcasync2} satisfy the conditions described in Section \ref{I-sec:model} of Part I \cite{Zhao13TSPasync1}; and the cost functions $\{J_k(w)\}$ satisfy Assumptions \ref{I-asm:costfunctions} and \ref{I-asm:boundedHessian} in Section \ref{I-sec:pre} of Part I \cite{Zhao13TSPasync1}. However, in order to study the mean-square-error performance in \emph{steady-state}, it is necessary to strengthen the assumption on the stochastic gradient vectors $\{\wh{\nabla_{w^*}J_k}(\bm{w}_{k,i-1})\}$. We replace the gradient noise model described in Assumption \ref{I-asm:gradientnoise} in Section \ref{I-subsec:syncDiff} of Part I \cite{Zhao13TSPasync1} by the following one. 

\begin{assumption}[Gradient noise model]
\label{asm:gradientnoisestronger}
\hfill
\begin{enumerate}
\item The gradient noise $\bm{v}_{k,i}(\bm{w}_{k,i-1})$, conditioned on $\F_{i-1}$, is assumed to be independent of any other random sources including topology, links, combination coefficients, and step-sizes. The conditional moments of $\bm{v}_{k,i}(\bm{w}_{k,i-1})$ satisfy:
\begin{align}
\label{eqn:vkicond_mean}
\E\,[ \bm{v}_{k,i}(\bm{w}_{k,i-1}) | \F_{i-1} ] &  =  0 \\
\label{eqn:vkicond_4thm}
\E\,[ \|\bm{v}_{k,i}(\bm{w}_{k,i-1}) \|^4 | \F_{i-1} ] &  \le  \alpha^2 \| w^o  -  \bm{w}_{k,i-1}\|^4  +  \sigma_v^4   
\end{align}
for some $\alpha \ge 0$ and $\sigma_v^2 \ge 0$.

\item The individual gradient noises $\{ \bm{v}_{k,i}(\bm{w}_{k,i-1}) \}$ are uncorrelated and circular across all agents such that
\be
\label{eqn:bigRiandRki}
\Rcal_i(\bm{w}_{i-1}) = \diag\{ R_{1,i}(\bm{w}_{1,i-1}), \dots, R_{N,i}(\bm{w}_{N,i-1}) \}
\ee
where $\Rcal_i(\bm{w}_{i-1})$ and $\{ R_{k,i}(\bm{w}_{k,i-1}) \}$ are from \eqref{I-eqn:bigRidef} and \eqref{I-eqn:Rkidef} both in Part I \cite{Zhao13TSPasync1}.

\item The conditional covariance of $\ubar{\bm{v}}_i(\bm{w}_{i-1})$ satisfies the Lipschitz condition
\be
\label{eqn:LipschitzRvi}
\| \Rcal_i(\one_N \kron w^o) - \Rcal_i(\bm{w}_{i-1}) \| \le \kappa_v \| \one_N \kron w^o - \bm{w}_{i-1} \|^{\gamma_v}
\ee
for some constants $\kappa_v \ge 0$ and $0 < \gamma_v \le 4$.

\item The covariance of $\ubar{\bm{v}}_i(\one_N \kron w^o)$ converges to a constant matrix:
\be
\label{eqn:bigRdef}
\Rcal \defeq \lim_{i\rightarrow\infty} \Rcal_i(\one_N \kron w^o) \defeq \diag\{R_1, \dots, R_N\}
\ee
where 
\be
\label{eqn:Rkdef}
R_k \defeq \lim_{i\rightarrow\infty} R_{k,i}(w^o)
\ee
\hfill \IEEEQED
\end{enumerate}
\end{assumption}

From Assumption \ref{asm:gradientnoisestronger}, the conditional moments of $\ubar{\bm{v}}_{k,i}(\bm{w}_{k,i-1})$ satisfy
\begin{align}
\label{eqn:gradientnoisemean}
\E[\ubar{\bm{v}}_{k,i}(\bm{w}_{k,i-1}) | \F_{i-1} ] & = 0 \\
\label{eqn:gradientnoise4thmoment}
\E[\|\ubar{\bm{v}}_{k,i}(\bm{w}_{k,i-1})\|^4 | \F_{i-1} ] & \le \alpha^2 \|\ubar{w}^o - \ubar{\bm{w}}_{k,i-1}\|^4 + 4 \sigma_v^4
\end{align}
where a factor of 4 appeared due to the transform $\ubar{\mbbT}(\cdot)$ from \eqref{I-eqn:ubardef} of Part I \cite{Zhao13TSPasync1}.

\subsection{Long Term Error Dynamics}
\label{subsec:longtermhighprob}
We showed in \eqref{I-eqn:errorrecursiondef} from Part I \cite{Zhao13TSPasync1} that the error recursion for the asynchronous network \eqref{eqn:atcasync1}--\eqref{eqn:atcasync2} evolves according to the following dynamics:
\be
\label{eqn:errorrecursiondef}
\ubar{\wt{\bm{w}}}_i  =  \bm{\Acal}_i^\T (I_{2MN}  -  \bm{\Mcal}_i\bm{\Hcal}_{i-1})\ubar{\wt{\bm{w}}}_{i-1}  +  \bm{\Acal}_i^\T\bm{\Mcal}_i\ubar{\bm{v}}_i(\bm{w}_{i-1})
\ee
where $\bm{\Hcal}_{i-1} = \diag\{ \bm{H}_{1,i-1}, \bm{H}_{2,i-1}, \dots, \bm{H}_{N,i-1} \}$ and 
\be
\label{eqn:Hkidef}
\bm{H}_{k,i-1} \defeq \int_{0}^{1} \nabla_{\ubar{w}\ubar{w}^*}^2 J_k(\ubar{w}^o-t\ubar{\wt{\bm{w}}}_{k,i-1})\,dt
\ee
The dependency of $\bm{\Hcal}_{i-1}$ on the previous iterate $\bm{w}_{i-1}$ complicates the mean-square analysis. Recall though from Lemma \ref{I-lemma:lipschitzHessian} in Part I \cite{Zhao13TSPasync1} that the Hessian matrices of the costs $\{J_k(\ubar{w})\}$ are globally Lipschitz around $\ubar{w}^o$. Let
\be
\label{eqn:Hkdef}
H_k \defeq \nabla_{\ubar{w}\ubar{w}^{*}}^2 J_k(\ubar{w}^o), \quad \Hcal \defeq \diag\{H_1,\dots,H_N\}
\ee
Recursion \eqref{eqn:errorrecursiondef} can then be rewritten as
\be
\label{eqn:errorrecursiondefdelta}
\ubar{\wt{\bm{w}}}_i = \bm{\Acal}_i^\T (I_{2MN}   -  \bm{\Mcal}_i \Hcal) \ubar{\wt{\bm{w}}}_{i-1} + \bm{\Acal}_i^\T \bm{\Mcal}_i \ubar{\bm{v}}_i(\bm{w}_{i-1}) + \bm{\Acal}_i^\T \bm{d}_i
\ee
where the perturbation factor $\bm{d}_i$ is given by
\begin{align}
\label{eqn:Deltaidef}
& \bm{d}_i \defeq \bm{\Mcal}_i (\Hcal  -  \bm{\Hcal}_{i-1}) \ubar{\wt{\bm{w}}}_{i-1}  \defeq  \col\{ \bm{d}_{1,i}, \dots, \bm{d}_{N,i} \} \\
\label{eqn:Deltakidef}
& \bm{d}_{k,i} \defeq \bm{\mu}_k(i) (H_k - \bm{H}_{k,i-1}) \ubar{\wt{\bm{w}}}_{k,i-1}
\end{align}
Let $\bar{\mu}_k^{(n)} \defeq \E[ \bm{\mu}_k(i) ]^n$ denote the $n$-th moment of the random step-size parameter $\bm{\mu}_k(i)$; we also use $\bar{\mu}_k \equiv \bar{\mu}_k^{(1)}$ from \eqref{I-eqn:randstepsizemeanentry} of Part I \cite{Zhao13TSPasync1} for the mean and $c_{\mu,k,\ell} = \E [(\bm{\mu}_k(i) - \bar{\mu}_k)(\bm{\mu}_\ell(i) - \bar{\mu}_\ell)]$ from \eqref{I-eqn:randstepsizecoventry} of Part I \cite{Zhao13TSPasync1} for the cross-covariance.

\begin{lemma}[Size of perturbation]
\label{lemma:sizeofdki}
If condition \eqref{I-eqn:boundstepsize4thorder} in Part I \cite{Zhao13TSPasync1}, namely,
\be
\label{eqn:boundstepsize4thorder}
\frac{\sqrt{\bar{\mu}_k^{(4)}}}{\bar{\mu}_k^{(1)} } < \frac{ \lambda_{k, \min} }{ 3\lambda_{k, \max}^2 + 4\alpha }
\ee
holds for all $k$, then 
\be
\label{eqn:bounddkimeansquare}
\limsup_{i \rightarrow \infty} \E \| \bm{\Acal}_i^\T \bm{d}_i \|^2 \le O(\nu^4)
\ee
where
\be
\label{eqn:nudef}
\nu \defeq \max_{k} \frac{\sqrt{\bar{\mu}_k^{(4)}}}{\bar{\mu}_k^{(1)} }
\ee
\end{lemma}

\begin{proof}
See Appendix \ref{app:sizeofdki}.
\end{proof}

\begin{assumption}[{Small step-sizes}]
\label{asm:smallstepsizes} 
The parameter $\nu$ from \eqref{eqn:nudef} is sufficiently small such that
\be
\nu < \min_k \frac{ \lambda_{k, \min} }{ 3\lambda_{k, \max}^2 + 4\alpha } < 1
\ee
\hfill \IEEEQED
\end{assumption}

Under Assumption \ref{asm:smallstepsizes}, condition \eqref{eqn:boundstepsize4thorder} holds. It was shown in \eqref{I-eqn:oldnulessnewnu} and \eqref{I-eqn:newboundlessoldbound} from Part I \cite{Zhao13TSPasync1} that condition \eqref{eqn:boundstepsize4thorder} in this Part implies condition \eqref{I-eqn:meansquarestabilitycond3} from Part I \cite{Zhao13TSPasync1}, i.e.,
\be
\label{eqn:meansquarestabilitycond3} 
\frac{\bar{\mu}_k^{(2)}}{\bar{\mu}_k^{(1)}} < \frac{\lambda_{k,\min}}{\alpha + \lambda_{k,\max}^2}
\ee
for all $k$.

Since we are interested in examining the asymptotic performance of the asynchronous network, result \eqref{eqn:bounddkimeansquare} indicates that the network error recursion \eqref{eqn:errorrecursiondef} can be expressed for large enough $i$ by using the following \emph{long-term} model:
\be
\label{eqn:errorrecursiondefapprox}
\ubar{\wt{\bm{w}}}_i' = \bm{\Acal}_i^\T (I_{2MN} - \bm{\Mcal}_i\Hcal)\ubar{\wt{\bm{w}}}_{i-1}' + \bm{\Acal}_i^\T \bm{\Mcal}_i \ubar{\bm{v}}_i(\bm{w}_{i-1})
\ee
where we ignore the $O(\nu^2)$ term $\bm{\Acal}_i^\T \bm{d}_i$ according to \eqref{eqn:bounddkimeansquare}, and we use $\bm{w}_{i-1}'$ to denote the estimate obtained from this long-term model. It is worth noting that the gradient noise $\ubar{\bm{v}}_i(\bm{w}_{i-1})$ in \eqref{eqn:errorrecursiondefapprox} is an extraneous noise that is imported from the original model \eqref{eqn:errorrecursiondef}; it only depends on the original estimate $\bm{w}_{i-1}$ but not on $\bm{w}_{i-1}'$. We will now use recursion \eqref{eqn:errorrecursiondefapprox} to determine expressions (rather than bounds) for the steady-state individual MSD and for the average network MSD. One advantage of model \eqref{eqn:errorrecursiondefapprox} is that the random matrix $\bm{\Hcal}_{i-1}$ from \eqref{eqn:errorrecursiondef} has been replaced by the constant matrix $\Hcal$. More formally, under Assumption \ref{asm:gradientnoisestronger} on the fourth-order moment of the gradient noise, and by extending the arguments of Appendices \ref{I-app:proofofstablity} and \ref{I-app:4thorder} from Part I \cite{Zhao13TSPasync1} and the arguments of \cite{Chen13TIT}, we will establish later in \eqref{eqn:networkMSDproof3} that the MSD expression resulting from \eqref{eqn:errorrecursiondefapprox} is within $O(\nu^{3/2})$ of the MSD expression for the original recursion \eqref{eqn:errorrecursiondef}; this conclusion will rely on the following useful result.

\begin{theorem}[Bounded mean-square gap]
\label{theorem:gapmeansquare}
Under Assumptions \ref{asm:gradientnoisestronger} and \ref{asm:smallstepsizes}, the mean-square gap from the original error recursion \eqref{eqn:errorrecursiondef} to the long-term model \eqref{eqn:errorrecursiondefapprox} is then asymptotically bounded by 
\be
\limsup_{i\rightarrow\infty} \left[ \max_k \E \|\ubar{\wt{\bm{w}}}_{k,i} - \ubar{\wt{\bm{w}}}_{k,i}' \|^2 \right] \le O(\nu^2)
\ee
for any $k$.
\end{theorem}
\begin{IEEEproof}
See Appendix \ref{app:boundgapmeansquare}.
\end{IEEEproof}

To proceed with the mean-square-error performance analysis, we introduce the following auxiliary variables:
\begin{align}
\label{eqn:Dkidef}
\bm{D}_{k,i} & \defeq I_{2M} - \bm{\mu}_k(i) H_k \\
\label{eqn:bigDidef}
\bm{\Dcal}_i & \defeq I_{2MN} - \bm{\Mcal}_i \Hcal = \diag\{ \bm{D}_{k,i} \} \\
\label{eqn:bigBidef}
\bm{\Bcal}_i & \defeq \bm{\Acal}_i^\T \bm{\Dcal}_i \\
\label{eqn:bigsidef}
\ubar{\bm{s}}_i & \defeq \bm{\Acal}_i^\T \bm{\Mcal}_i \ubar{\bm{v}}_i(\bm{w}_{i-1})
\end{align}
Based on the gradient noise model in Assumption \ref{asm:gradientnoisestronger} and the asynchronous network model described in Section \ref{I-sec:model} of Part I \cite{Zhao13TSPasync1}, it is easy to verify that the (conditional) means of $\{\bm{\Acal}_i, \bm{\Mcal}_i, \bm{D}_{k,i}, \bm{\Dcal}_i, \bm{\Bcal}_i, \ubar{\bm{s}}_i\}$ are given by:
\begin{align}
\label{eqn:bigAmeandef}
\bar{\Acal} & \defeq \E(\bm{\Acal}_i) = \bar{A} \kron I_{2M} \\
\label{eqn:bigMmeandef}
\bar{\Mcal} & \defeq \E(\bm{\Mcal}_i) = \bar{M} \kron I_{2M} \\
\label{eqn:Dkmeandef}
\bar{D}_k & \defeq \E(\bm{D}_{k,i}) = I_{2M} - \bar{\mu}_k H_k \\
\label{eqn:bigDmeandef}
\bar{\Dcal} & \defeq \E(\bm{\Dcal}_i) = I_{2MN} - \bar{\Mcal} \Hcal = \diag\{ \bar{D}_k \} \\
\label{eqn:bigBmeandef}
\bar{\Bcal} & \defeq \E(\bm{\Bcal}_i) = \bar{\Acal}^\T \bar{\Dcal} \\
\label{eqn:bigsmeandef}
\bar{s} & \defeq \E(\ubar{\bm{s}}_i | \F_{i-1} ) = 0
\end{align}
It can be verified that the block-Kronecker-covariance matrices of several random quantities are given by:
\begin{align}
\label{eqn:bigCAdef}
\Ccal_A & \defeq \E[(\bm{\Acal}_i - \bar{\Acal}) \bkron (\bm{\Acal}_i - \bar{\Acal})] = C_A \kron I_{4M^2} \\
\label{eqn:bigCMdef}
\Ccal_M & \defeq \E[(\bm{\Mcal}_i - \bar{\Mcal}) \bkron (\bm{\Mcal}_i - \bar{\Mcal}) ] = C_M \kron I_{4M^2} \\
\label{eqn:bigCDdef}
\Ccal_D & \defeq \E[(\bm{\Dcal}_i^* - \bar{\Dcal}^*)^\T \bkron (\bm{\Dcal}_i - \bar{\Dcal}) ] = \Ccal_M(\Hcal^\T\bkron\Hcal) \\
\label{eqn:bigCBdef}
\Ccal_B & \defeq \E[(\bm{\Bcal}_i^* - \bar{\Bcal}^*)^\T \bkron (\bm{\Bcal}_i - \bar{\Bcal}) ] = (\bar{\Acal}^\T \bkron \bar{\Acal}^\T) \Ccal_D + \Ccal_A^\T (\bar{\Dcal}^\T \bkron \bar{\Dcal} + \Ccal_D)
\end{align}
where the symbol $\bkron$ denotes the block-Kronecker operation of block size $2M \times 2M$ (see Appendix \ref{app:blockvecandkronecker}). Moreover, it can be verified by using property \eqref{eqn:bvecadd} from Appendix \ref{app:blockvecandkronecker} that
\begin{align}
\label{eqn:blockkroneckercovariance}
\E[(\bm{\Xcal}^* - \bar{\Xcal}^*)^\T \bkron (\bm{\Xcal} - \bar{\Xcal})] = \E[(\bm{\Xcal}^*)^\T \bkron  \bm{\Xcal}] - (\bar{\Xcal}^*)^\T \bkron \bar{\Xcal}
\end{align}
for any random block matrix $\bm{\Xcal}$ with appropriate block size and with mean $\bar{\Xcal} \defeq \E\bm{\Xcal}$. The $\{C_A, C_M\}$ that appear in \eqref{eqn:bigCAdef}--\eqref{eqn:bigCBdef} relate to the second-order moments of $\{\bm{a}_{\ell k}(i)\}$ and $\{\bm{\mu}_k(i)\}$. Using \eqref{eqn:bigBidef} and  \eqref{eqn:bigsidef}, the long-term model \eqref{eqn:errorrecursiondefapprox} can be rewritten as
\be
\label{eqn:errorrecursiondefapprox_new}
\ubar{\wt{\bm{w}}}_i' = \bm{\Bcal}_i \cdot \ubar{\wt{\bm{w}}}_{i-1}' + \ubar{\bm{s}}_i
\ee

\subsection{Mean Error Recursion}
Taking the expectation of both sides of \eqref{eqn:errorrecursiondefapprox_new}, we end up with the mean error recursion for large $i$:
\be
\label{eqn:meanerrorrecursion}
\E\,\ubar{\wt{\bm{w}}}_i' = \bar{\Bcal} \cdot \E\,\ubar{\wt{\bm{w}}}_{i-1}'
\ee
The stability of recursion \eqref{eqn:meanerrorrecursion} requires the stability of $\bar{\Bcal}$. A condition on the step-sizes to ensure the stability of $\bar{\Bcal}$ can be derived as follows. Using the fact that $\bar{\Acal}$ is block left-stochastic and $\bar{\Dcal}$ is block diagonal and Hermitian, and following the same argument in \cite[App.~A]{Zhao12TSP} \cite{Sayed13Chapter}, we obtain
\be
\label{eqn:Bcalbinfnorm}
\rho(\bar{\Bcal}) \le \rho(\bar{\Dcal})
\ee
where $\rho(\cdot)$ denotes the spectral radius of its matrix argument. It follows from \eqref{eqn:bigDmeandef} and \eqref{eqn:Bcalbinfnorm} that asymptotic mean stability is guaranteed if the mean step-size $\bar{\mu}_k$ satisfies 
\be
\label{eqn:rhobigBlessthan1}
\bar{\mu}_k \equiv \bar{\mu}_k^{(1)} < \frac{2}{\rho(H_k)}
\ee
for all $k$. Since $H_k$ is a positive semi-definite matrix, its spectral radius coincides with its largest eigenvalue. Using \eqref{I-eqn:boundseigHessian} from Part I \cite{Zhao13TSPasync1}, we have $\rho(H_k) \le \lambda_{k,\max}$. If condition \eqref{eqn:boundstepsize4thorder} holds, then from \eqref{I-eqn:boundmumoments} of Part I \cite{Zhao13TSPasync1}, we have
\be
\label{eqn:meanstabecondition}
\bar{\mu}_k^{(1)} \le \frac{\sqrt{\bar{\mu}_k^{(4)}}}{\bar{\mu}_k^{(1)}} < \frac{\lambda_{k,\min}}{ 3 \lambda_{k,\max}^2 + 4 \alpha } \le \frac{\lambda_{k,\min}}{ 3 \lambda_{k,\max}^2 } \le \frac{2}{\rho(H_k)}
\ee
since $\alpha > 0$. Therefore, condition \eqref{eqn:rhobigBlessthan1} holds if condition \eqref{eqn:boundstepsize4thorder} does so. With Assumption \ref{asm:smallstepsizes}, we have
\be
\label{eqn:limmeanerror}
\lim_{i\rightarrow\infty} \E\,\wt{\bm{w}}_{k,i}' = 0
\ee
for all $k$. From \eqref{eqn:limmeanerror}, we conclude that the long-term model \eqref{eqn:errorrecursiondefapprox} or, equivalently, \eqref{eqn:errorrecursiondefapprox_new}, is the asymptotically centered version of the original error recursion \eqref{eqn:errorrecursiondef}.

\subsection{Error Covariance Recursion}
\label{sec:errorcovrecursion}
We proceed to examine the evolution of the covariance matrix of the  network error vector $\ubar{\wt{\bm{w}}}_i'$ in the long-term model \eqref{eqn:errorrecursiondefapprox_new}. Let
\begin{align} 
\label{eqn:bigridef}
r_i(\bm{w}_{i-1}) & \defeq \bvec( \Rcal_i(\bm{w}_{i-1}) ) 
= \E [ (\ubar{\bm{v}}_i^*(\bm{w}_{i-1}))^\T \bkron \ubar{\bm{v}}_i(\bm{w}_{i-1}) | \F_{i-1} ] \\
\label{eqn:bigyidef}
y_i & \defeq (\bar{\Acal}^\T \bkron \bar{\Acal}^\T + \Ccal_A^\T)(\bar{\Mcal} \bkron \bar{\Mcal} + \Ccal_M) \E [r_i(\bm{w}_{i-1})] \\
\label{eqn:bigzidef}
z_i & \defeq \bvec( \E\,(\ubar{\wt{\bm{w}}}_i' \ubar{\wt{\bm{w}}}_i'^*) ) = \E\,[(\ubar{\wt{\bm{w}}}_i'^*)^\T \bkron \ubar{\wt{\bm{w}}}_i' ] \\
\label{eqn:bigGdef}
\Gcal & \defeq \E[(\bm{\Dcal}_i^*)^\T \bkron \bm{\Dcal}_i] = \bar{\Dcal}^\T \bkron \bar{\Dcal} + \Ccal_D \\
\label{eqn:bigFdef}
\Fcal & \defeq \E[(\bm{\Bcal}_i^*)^\T \bkron \bm{\Bcal}_i]^* = \bar{\Bcal}^\T \bkron \bar{\Bcal}^* + \Ccal_B^* = \Gcal(\bar{\Acal} \bkron \bar{\Acal} + \Ccal_A)
\end{align}
where the notations $\bvec(\cdot)$ and $\bkron$ denote block vectorization and block Kronecker products, respectively, both of size $2M\times 2M$ (see Appendix \ref{app:blockvecandkronecker}). We note that the second equalities in \eqref{eqn:bigridef} and \eqref{eqn:bigzidef} are due to property \eqref{eqn:bvecproduct_vec} and the second equalities in \eqref{eqn:bigGdef} and \eqref{eqn:bigFdef} are by using \eqref{eqn:bigDmeandef}, \eqref{eqn:bigBmeandef}, and \eqref{eqn:bigCDdef}--\eqref{eqn:blockkroneckercovariance}. Using \eqref{eqn:bigridef}--\eqref{eqn:bigFdef}, we obtain the following recursion for the block-vectorized covariance matrix of the network error vector $\ubar{\wt{\bm{w}}}_i'$.

\begin{theorem}[{Network error covariance recursion}]
\label{theorem:errorcovariancerecursion}
The vector $z_i$ evolves according to the following recursion:
\be
\label{eqn:meansquareerrorrecursion}
z_i = \Fcal^* z_{i-1} + y_i
\ee
Recursion \eqref{eqn:meansquareerrorrecursion} converges if condition \eqref{eqn:boundstepsize4thorder} holds, and its convergence rate is determined by $\rho(\Fcal)$.
\end{theorem}
\begin{IEEEproof}
See Appendix \ref{app:errorcovariancerecursion}.
\end{IEEEproof}

The vector $z_i$ can be used to compute useful error metrics. For example, we can examine any weighted MSE measure for $\ubar{\wt{\bm{w}}}_i'$ by evaluating quantities of the form:
\be
\label{eqn:weighted2norm}
\E\,\|\ubar{\wt{\bm{w}}}_i'\|_\Sigma^2 = \E\,[\Tr(\ubar{\wt{\bm{w}}}_i' \ubar{\wt{\bm{w}}}_i'^* \Sigma)] = z_i^* \cdot \bvec(\Sigma)
\ee
where $\Sigma$ is an arbitrary positive semi-definite weight matrix. To guarantee the convergence of $\E\,\|\ubar{\wt{\bm{w}}}_i'\|_\Sigma^2$ for any weighting matrix $\Sigma$, it is sufficient and necessary to guarantee the convergence of $z_i$. It follows from Theorem \ref{theorem:errorcovariancerecursion} that under Assumption \ref{asm:smallstepsizes}, the spectral radius of the matrix $\Fcal$ in \eqref{eqn:meansquareerrorrecursion} determines the mean-square stability and convergence rate of the asynchronous diffusion strategy \eqref{eqn:atcasync1}--\eqref{eqn:atcasync2}.

Before proceeding we comment on the reason why we choose to use the block vectorization operation $\bvec(\cdot)$ in \eqref{eqn:bigzidef} instead of the traditional vectorization operation $\vecm(\cdot)$. This is because $\bvec(\cdot)$ allows us to track each block of its matrix argument \emph{after} vectorization. By the definition in \eqref{eqn:bvecdef} and the illustration in Fig.~\ref{fig:vectorization}, operation $\bvec(\cdot)$ preserves the locality of every block in the original matrix argument whereas operation $\vecm(\cdot)$ blends different blocks together. Therefore, whenever we need to vectorize a network matrix whose blocks relate to individual agents, it is more natural to use the block vectorization operation $\bvec(\cdot)$; on the other hand, whenever we need to vectorize a matrix that only relates to a single agent, we can use the conventional vectorization operation $\vecm(\cdot)$. A useful property of the conventional vectorization operation $\vecm(\cdot)$ is
\be
\vecm(ABC) = (C^\T \kron A)\cdot \vecm(B)
\ee 
for matrices $\{A,B,C\}$ of compatible sizes. A similar property holds for the $\bvec(\cdot)$ operation:
\be
\bvec(ABC) = (C^\T \bkron A) \cdot \bvec(B)
\ee
for block matrices $\{A,B,C\}$ with appropriate block sizes. In Fig.~\ref{fig:kronecker}, we compare the structures of $A \kron B$ and $A \bkron B$, where $\{A,B\}$ are a pair of block matrices. The observation is that the operation $\kron$ destroys the locality of the blocks from matrix $B$, whereas the operation $\bkron$ preserves the locality of the blocks from both matrices $A$ and $B$.

Using properties of the block operations $\bvec(\cdot)$ and $\bkron$, we can derive from Theorem \ref{theorem:errorcovariancerecursion} a useful relation between the blocks of the network error covariance matrix, $\E \ubar{\wt{\bm{w}}}_i' \ubar{\wt{\bm{w}}}_i'^* $, and the blocks of the vector $z_i$. Let us partition the $4M^2N^2$-dimensional vector $z_i$ as
\be
\label{eqn:ziellblockexpression}
z_i = \col\{ z_i^{(1)}, \dots, z_i^{(N)} \}, \qquad
z_i^{(\ell)} \defeq \col\{ z_i^{(\ell,1)}, \dots, z_i^{(\ell,N)} \}   
\ee
where $z_i^{(\ell)}$ is the $\ell$-th sub-vector of $z_i$ with dimension $4M^2N$ and $z_i^{(\ell,k)}$ is the $k$-th block of $z_i^{(\ell)}$ with dimension $4M^2$. From \eqref{eqn:bigzidef} and \eqref{eqn:bvecdef}, we find that these vectors have the following useful interpretations for $k, \ell = 1, 2, \dots, N$:
\begin{align}
\label{eqn:siblock}
& z_i = \E [\bvec (\ubar{\wt{\bm{w}}}_i' \ubar{\wt{\bm{w}}}_i'^* ) ] = \col\{ \E\,[ (\ubar{\wt{\bm{w}}}_{\ell,i}'^*)^\T \kron \ubar{\wt{\bm{w}}}_{k,i}' ] \}_{\ell,k=1}^{N} \\
\label{eqn:zilkdef}
& z_i^{(\ell,k)} \defeq \vecm( \E[\ubar{\wt{\bm{w}}}_{k,i}' \ubar{\wt{\bm{w}}}_{\ell,i}'^*] ) = \E[ (\ubar{\wt{\bm{w}}}_{\ell,i}'^*)^\T \kron \ubar{\wt{\bm{w}}}_{k,i}' ]
\end{align}
where $\E \ubar{\wt{\bm{w}}}_{k,i}' \ubar{\wt{\bm{w}}}_{\ell,i}'^* $ is the $(k,\ell)$-th block of $\E \ubar{\wt{\bm{w}}}_i' \ubar{\wt{\bm{w}}}_i'^* $ with size $2M \times 2M$. The block entries of the vector $z_i$ in \eqref{eqn:zilkdef} do not only allow us to recover the covariance matrices of any individual error vectors, $\E \ubar{\wt{\bm{w}}}_{k,i}' \ubar{\wt{\bm{w}}}_{k,i}'^* $, but they also allow us to recover the \emph{cross-covariance} matrices, $\E \ubar{\wt{\bm{w}}}_{k,i}' \ubar{\wt{\bm{w}}}_{\ell,i}'^* $, for \emph{any} pair of agents $\{k,\ell\}$. Therefore, by studying the evolution of the entire covariance vector in \eqref{eqn:meansquareerrorrecursion}, we are able to extract some detailed information about the dynamics of the asynchronous diffusion network, as we shall show in Theorem \ref{theorem:networkcovariance} and Corollary \ref{corollary:clustered} in Section \ref{sec:lowrank}. 

\begin{figure}
\includegraphics[scale=0.6]{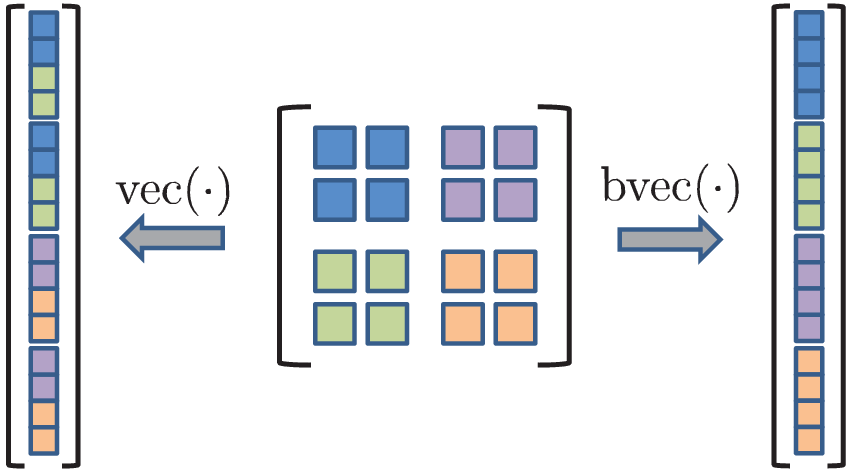}
\centering
\caption{Comparing two vectorization operations: $\vecm(\cdot)$ versus $\bvec(\cdot)$. The operation $\vecm(\cdot)$ destroys the locality of the blocks in the original matrix argument while the operation $\bvec(\cdot)$ preserves it.}
\label{fig:vectorization}
\vspace{-1\baselineskip}
\end{figure}

\begin{figure}
\includegraphics[scale=0.7]{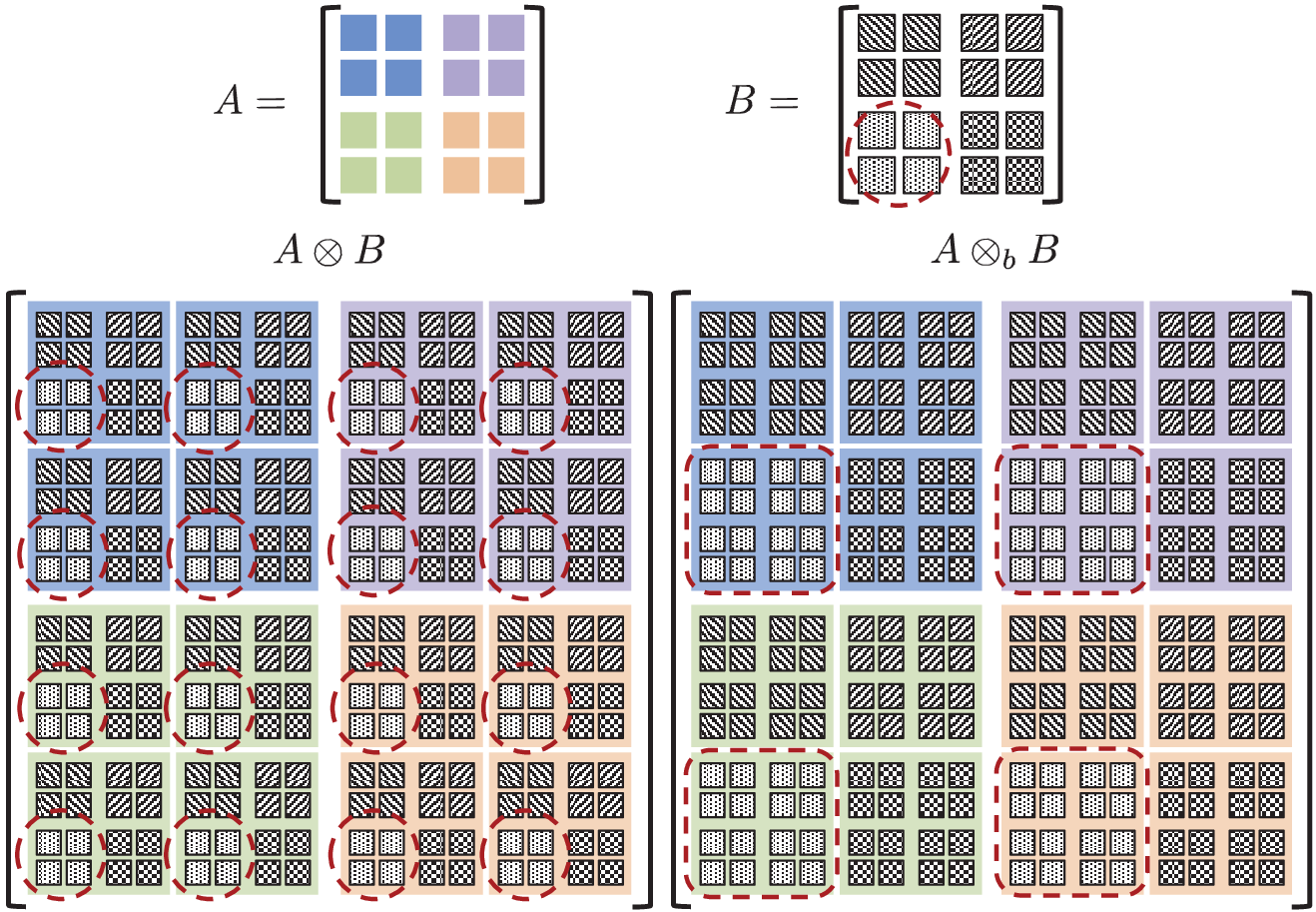}
\centering
\caption{Comparing two Kronecker product operations: $\kron$ versus $\bkron$. The operation $\kron$ destroys the locality of the blocks from matrix $B$ while the operation $\bkron$ preserves the locality of the blocks from both matrices $A$ and $B$.}
\label{fig:kronecker}
\vspace{-1\baselineskip}
\end{figure}

\section{Steady-State Performance}
When $i \rightarrow \infty$, and by the fact that $\Fcal$ is stable, we obtain from \eqref{eqn:meansquareerrorrecursion} that
\begin{align}
\label{eqn:sinftydef}
z_{\infty} & \defeq \lim_{i\rightarrow\infty} z_i = (I_{4M^2N^2} - \Fcal^*)^{-1} \lim_{i\rightarrow\infty} y_i \nn \\
& = (I_{4M^2N^2} - \Fcal^*)^{-1} (\bar{\Acal}^\T \bkron \bar{\Acal}^\T + \Ccal_A^\T) (\bar{\Mcal} \bkron \bar{\Mcal} + \Ccal_M) \lim_{i\rightarrow\infty} \bvec( \E \Rcal_i(\bm{w}_{i-1}) )
\end{align}
where we also used \eqref{eqn:bigridef} and \eqref{eqn:bigyidef}. Now note that
\begin{align}
\label{eqn:bounddiffRi}
\| \Rcal_i(\one_N \kron w^o) - \E \Rcal_i(\bm{w}_{i-1}) \| 
& \stackrel{(a)}{\le} \E \| \Rcal_i(\one_N \kron w^o) - \Rcal_i(\bm{w}_{i-1}) \| \nn \\
& \stackrel{(b)}{\le} \kappa_v \cdot \E \, \| \one_N \kron w^o - \bm{w}_{i-1} \|^{\gamma_v} \nn \\
& = \kappa_v \cdot \E \, [ \| \wt{\bm{w}}_{i-1} \|^4 ]^{\gamma_v/4} \nn \\
& \stackrel{(c)}{\le} \kappa_v \cdot [ \E \, \| \wt{\bm{w}}_{i-1} \|^4 ]^{\gamma_v/4}
\end{align}
where step (a) is by Jensen's inequality; step (b) is by \eqref{eqn:LipschitzRvi} in Assumption \ref{asm:gradientnoisestronger}; and step (c) is by Jensen's inequality and the fact that $| \cdot |^{\gamma_v/4}$ is concave due to $0 < \gamma_v/4 \le 1$. Now, from Theorem \ref{I-theorem:4thmoments} in Part I \cite{Zhao13TSPasync1}, we know that $\limsup_{i\rightarrow\infty} \E \, \| \wt{\bm{w}}_{i-1} \|^4 \le O(\nu^2)$ under Assumption \ref{asm:smallstepsizes}. Therefore, we obtain from \eqref{eqn:bounddiffRi} that
\be
\label{eqn:bounddiffRi1}
\limsup_{i\rightarrow\infty} \| \Rcal_i(\one_N \kron w^o) - \E \Rcal_i(\bm{w}_{i-1}) \| \le O(\nu^{\gamma_v/2})
\ee
According to \eqref{eqn:bounddiffRi1}, we can replace $\E \Rcal_i(\bm{w}_{i-1})$ in \eqref{eqn:sinftydef} by $\Rcal_i(\one_N \kron w^o)$ with an error in the order of $\nu^{\gamma_v/2}$. Let
\begin{align}
\label{eqn:bigzdef}
z \defeq (I_{4M^2N^2} - \Fcal^*)^{-1} (\bar{\Acal}^\T \bkron \bar{\Acal}^\T + \Ccal_A^\T) (\bar{\Mcal} \bkron \bar{\Mcal} + \Ccal_M) \bvec( \Rcal ) 
\end{align}
From \eqref{I-eqn:boundmumoments2} in Part I \cite{Zhao13TSPasync1}, we know that the second-order moments of $\{ \bm{\mu}_k(i) \}$ are in the order of $\nu^2$. Hence, by \eqref{I-eqn:bigMidef} from Part I \cite{Zhao13TSPasync1}, \eqref{eqn:bigMmeandef}, and \eqref{eqn:bigCMdef}, it is easy to verify that 
\be
\label{eqn:boundMMCM}
\| \bar{\Mcal} \bkron \bar{\Mcal} + \Ccal_M \| = O(\nu^2)
\ee
Using \eqref{eqn:bigzdef}, \eqref{eqn:boundMMCM}, and the fact that $\| (I_{4M^2N^2} - \Fcal^*)^{-1} \| = O(\nu^{-1})$ from Lemma \ref{lemma:lowrank} further ahead, we conclude that
\be
\label{eqn:bigzorder}
\| z \| = O(\nu)
\ee
Then, by using \eqref{eqn:bigRdef} and \eqref{eqn:bounddiffRi1}--\eqref{eqn:bigzorder}, we obtain from \eqref{eqn:sinftydef} that 
\be
\label{eqn:zinftyandz}
z_{\infty} = z + O(\nu^{1 + \gamma_v/2}), \qquad \| z_{\infty} \| = O(\nu)
\ee
Define the steady-state average network MSD by
\be
\label{eqn:networkMSDdef}
\MSD^{\network} \defeq \lim_{i\rightarrow\infty} \frac{1}{N} \sum_{k=1}^{N} \E\,\|\wt{\bm{w}}_{k,i}\|^2
\ee
and the steady-state individual MSD for agent $k$ by
\be
\label{eqn:individualMSDdef}
\MSD_k \defeq \lim_{i\rightarrow\infty} \E\,\|\wt{\bm{w}}_{k,i}\|^2
\ee

\begin{theorem}[Steady-state MSD]
\label{theorem:steadystateMSD}
It holds that
\begin{align}
\label{eqn:networkMSDexpression}
\MSD^{\network} & = \frac{1}{2N} z^* \bvec(I_{2MN}) + O(\nu^{1 + \gamma_o}) \\
\label{eqn:individualMSDexpression}
\MSD_k & = \frac{1}{2} z^* \bvec( E_{kk} \kron I_{2M} ) + O(\nu^{1 + \gamma_o })
\end{align}
where $z$ is given by \eqref{eqn:bigzdef}, 
\be
\label{eqn:gammaodef}
\gamma_o \defeq \frac{1}{2} \min\{ 1, \gamma_v \}
\ee
and $E_{kk}$ is the $N\times N$ basis matrix that only has one non-zero element, which is equal to $1$, at the $(k,k)$-th entry.
\end{theorem}

\begin{IEEEproof}
From \eqref{eqn:weighted2norm} by selecting $\Sigma = I_{2MN}$, and also using \eqref{eqn:sinftydef} and \eqref{eqn:zinftyandz}, we get
\begin{align}
\label{eqn:networkMSDproof}
\lim_{i\rightarrow\infty} \E \|\ubar{\wt{\bm{w}}}_i'\|^2 & = z_\infty^* \bvec(I_{2MN}) \nn \\
& = z^* \bvec(I_{2MN}) + O(\nu^{1 + \gamma_v/2}) \nn \\
& = O(\nu)
\end{align}
Likewise, by selecting $\Sigma = E_{kk} \kron I_{2M}$, we get
\begin{align}
\label{eqn:individualMSDproof}
\lim_{i\rightarrow\infty} \E\|\ubar{\wt{\bm{w}}}_i'\|_{E_{kk} \kron I_{2M}}^2 & = z_\infty^* \bvec( E_{kk} \kron I_{2M} ) \nn \\
& = z^* \bvec( E_{kk} \kron I_{2M} ) + O(\nu^{1 + \gamma_v/2}) \nn \\
& = O(\nu)
\end{align}
Note further that
\be
\label{eqn:networkMSDproof2}
\E \|\ubar{\wt{\bm{w}}}_i\|^2 = \E \| \ubar{\wt{\bm{w}}}_i' \|^2 + \E \|\ubar{\wt{\bm{w}}}_i - \ubar{\wt{\bm{w}}}_i'\|^2 + 2 \Re \E [ \ubar{\wt{\bm{w}}}_i'^* (\ubar{\wt{\bm{w}}}_i - \ubar{\wt{\bm{w}}}_i') ]
\ee
Using the Cauchy-Schwartz inequality, it can be verified that
\be
\label{eqn:boundRegap}
\left| \Re \E [ \ubar{\wt{\bm{w}}}_i'^* (\ubar{\wt{\bm{w}}}_i - \ubar{\wt{\bm{w}}}_i') ] \right| \le \sqrt{\E \| \ubar{\wt{\bm{w}}}_i' \|^2 \cdot \E \|\ubar{\wt{\bm{w}}}_i - \ubar{\wt{\bm{w}}}_i'\|^2 }
\ee
From Theorem \ref{theorem:gapmeansquare}, we have 
\be
\label{eqn:meansquaregap_network}
\lim_{i\rightarrow\infty} \E \| \ubar{\wt{\bm{w}}}_i - \ubar{\wt{\bm{w}}}_i' \|^2 \le O(\nu^2)
\ee
Substituting \eqref{eqn:networkMSDproof} and \eqref{eqn:meansquaregap_network} into \eqref{eqn:networkMSDproof2}, and using \eqref{eqn:boundRegap}, we get
\begin{align}
\label{eqn:networkMSDproof3}
\lim_{i\rightarrow\infty} \E \|\ubar{\wt{\bm{w}}}_i\|^2 & = \lim_{i\rightarrow\infty} \E \| \ubar{\wt{\bm{w}}}_i' \|^2 + O(\nu^2) + 2 \sqrt{O(\nu) \cdot O(\nu^2)} \nn \\
& = \lim_{i\rightarrow\infty} \E \| \ubar{\wt{\bm{w}}}_i' \|^2 + O(\nu^{3/2})
\end{align}
Results \eqref{eqn:networkMSDexpression} and \eqref{eqn:individualMSDexpression} follow from \eqref{eqn:networkMSDproof}, \eqref{eqn:individualMSDproof}, and \eqref{eqn:networkMSDproof3}.
\end{IEEEproof}

Result \eqref{eqn:networkMSDexpression} generalizes its counterpart (276) from \cite{Sayed13Chapter} for the synchronous diffusion strategy. Since expressions \eqref{eqn:networkMSDproof} and \eqref{eqn:individualMSDproof} are both related to the vector $z$ in \eqref{eqn:bigzdef}, let us examine $z$ more closely to reveal the implications of asynchronous adaptation and learning on performance.  Theorem \ref{theorem:networkcovariance} in the following section will lead to powerful alternative expressions for \eqref{eqn:networkMSDproof} and \eqref{eqn:individualMSDproof}. The new expressions will highlight some important properties about the behavior of the asynchronous network in steady-state, such as the behavior that was illustrated earlier in Fig. \ref{fig:normball}. The subsequent analysis relies on a useful low-rank factorization result.

\section{Low-Rank Factorization}
\label{sec:lowrank}
From \eqref{eqn:sinftydef} we see that the structure of $z$ depends on the structure of the matrix $(I_{4M^2N^2}-\Fcal^*)^{-1}$. In the following, we show that by retaining the dominant eigen-space of $(I_{4M^2N^2}-\Fcal^*)^{-1}$, we can obtain a more revealing MSD expression than \eqref{eqn:networkMSDexpression} that is still accurate to the order of $O(\nu^{1+\gamma_o})$.

\subsection{Perron Eigenvectors}
To proceed, we introduce the following condition on the matrix $\bar{A} \kron \bar{A} + C_A$.

\begin{assumption}[{Primitiveness of $\bar{A} \kron \bar{A} + C_A$}]
\label{asm:connected}
The matrix $\bar{A} \kron \bar{A} + C_A$ is assumed to be primitive \cite[p.~45]{BermanPF}, namely, that there exists a finite positive integer $j$ such that all entries of $(\bar{A} \kron \bar{A} + C_A)^j$ are positive. \hfill \IEEEQED
\end{assumption}

\begin{lemma}[{Primitiveness of $\bar{A}$}]
\label{lemma:primitiveA}
The matrix $\bar{A}$ is primitive if $\bar{A} \kron \bar{A} + C_A$ is primitive.
\end{lemma}
\begin{IEEEproof} 
See Appendix \ref{app:primitive}.
\end{IEEEproof}

Assumption \ref{asm:connected} is guaranteed if the directed graph (digraph) associated with the matrix $\bar{A}  \kron  \bar{A}  +  C_A$ is strongly-connected with as least one self-loop \cite[pp.~30,34]{BermanPF}. The digraph associated with $\bar{A}  \kron  \bar{A}  +  C_A$ is the union of all possible digraphs associated with the realizations of $\bm{A}_i  \kron  \bm{A}_i$ \cite[p.~29]{Bondy08}. Each possible digraph associated with $\bm{A}_i  \kron  \bm{A}_i$ is a Kronecker graph of order 2 generated by the initiator $\bm{A}_i$ \cite{Leskovec10JMLR}. Therefore, Assumption \ref{asm:connected} amounts to an assumption that the \emph{union} of all possible digraphs associated with the realizations of $\bm{A}_i  \kron  \bm{A}_i$ is strongly-connected with at least one self-loop. As illustrated in Fig.~\ref{fig:primitive}, this condition still allows the digraphs associated with $\bm{A}_i$ to be weakly-connected with or without self-loops or even to be disconnected. Important cases such as random gossip \cite{Boyd06TIT, Aysal09Allerton, Aysal09TSP, Jakovetic11TSP, Kar11TSP} or probabilistic diffusion \cite{Lopes08ICASSP, Takahashi10ICASSP} are therefore not ruled out by this condition. It can be verified that the converse of Lemma \ref{lemma:primitiveA} is generally not true: when the digraph associated with $\bar{A}$ is primitive, the digraph associated with $\bar{A} \kron \bar{A} + C_A$ does not even need to be connected. 

By Lemma \ref{I-lemma:leftstochastic} from Part I \cite{Zhao13TSPasync1} and the above Assumption \ref{asm:connected}, the matrix $\bar{A} \kron \bar{A} + C_A$ is left-stochastic and primitive. It follows from the Perron-Frobenius theorem \cite{BermanPF} \cite{Pillai05SPM} that this matrix has a unique eigenvalue at one and a pair of eigenvectors $\{\one_{N^2}, p\}$ with \emph{positive} entries satisfying:
\be
\label{eqn:pdef}
(\bar{A} \kron \bar{A} + C_A) \cdot p = p, \qquad  p^\T \cdot \one_{N^2} = 1
\ee
Likewise, the matrix $\bar{A}$ is also left-stochastic and primitive. It has a unique eigenvalue at one and a pair of eigenvectors $\{\one_N, \bar{p}\}$ with positive entries satisfying:
\be
\label{eqn:barpdef}
\bar{A} \cdot \bar{p} = \bar{p}, \qquad  \bar{p}^\T \cdot \one_{N} = 1
\ee
All other eigenvalues of $\bar{A} \kron \bar{A} + C_A$ and $\bar{A}$ are inside the unit circle. To simplify the presentation, we shall use the name ``Perron eigenvector'' to refer to the unique eigenvectors $p$ and $\bar{p}$ in the sequel. Since the vector $p$ is of dimension $N^2\times1$, we partition it into $N$ sub-vectors of dimension $N \times 1$ each:
\be
p = \col\{ p_1, p_2, \dots, p_N \}
\ee
where $p_k$ denotes the $k$-th sub-vector. We further define an $N \times N$ matrix $P_p$ whose columns are the sub-vectors $\{p_k\}$:
\be
\label{eqn:Ppdef}
P_p \defeq \unvecm(p) = \begin{bmatrix}
p_1 & p_2 & \dots & p_N
\end{bmatrix}
\ee
We use $p_{\ell,k}$ to denote the $(\ell,k)$-th element of matrix $P_p$, which is equal to the $\ell$-th element of $p_k$. 

\begin{figure}
\includegraphics[scale=0.5]{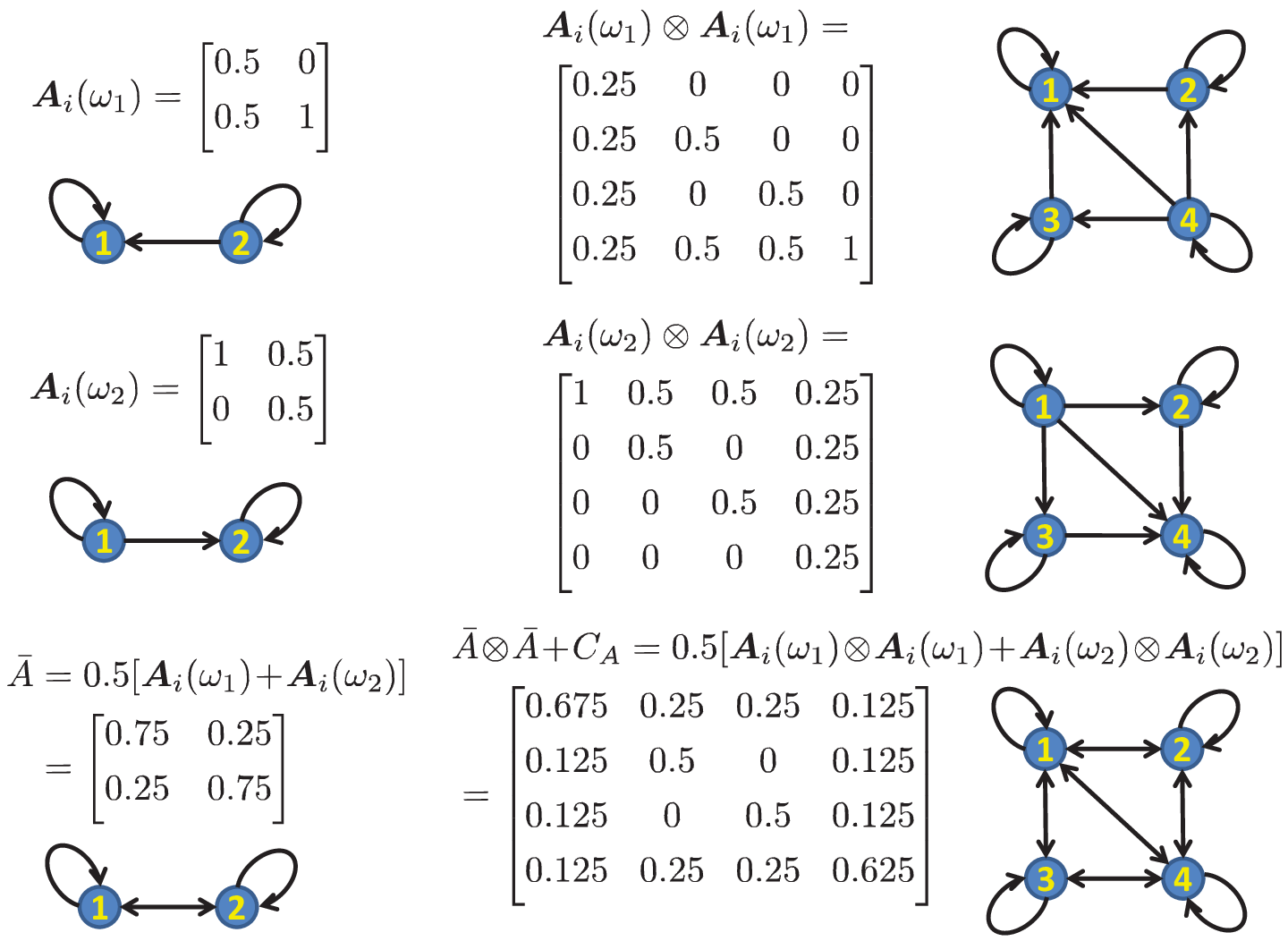}
\centering
\caption{An illustration of the digraph associated with $\E\,(\bm{A}_i \kron \bm{A}_i | \bm{w}_{i-1}) = \bar{A} \kron \bar{A} + C_A$, where $\bm{A}_i$ has two equally probable realizations $\{\bm{A}_i(\omega_1), \bm{A}_i(\omega_2)\}$. It can be observed that \emph{neither} of the digraphs associated with $\bm{A}_i(\omega_j) \otimes \bm{A}_i(\omega_j)$, $j=1,2$, is strongly-connected due to the existence of the source and sink nodes, where information can only flow in \emph{one} direction through the network. However, the digraph associated with $\E\,(\bm{A}_i \kron \bm{A}_i | \bm{w}_{i-1})$, which is the union of the first two digraphs, is strongly-connected, where information can flow in \emph{any} direction through the network.}
\label{fig:primitive}
\vspace{-1\baselineskip}
\end{figure}

\begin{lemma}[Properties of $P_p$]
\label{lemma:Ppandp}
The matrix $P_p$ in \eqref{eqn:Ppdef} is symmetric positive semi-definite and it satisfies $P_p \one_N = \bar{p}$, where the $\bar{p}$ is the Perron eigenvector in \eqref{eqn:barpdef}.
\end{lemma}
\begin{IEEEproof}
See Appendix \ref{app:Ppandp}.
\end{IEEEproof}

From Lemma \ref{lemma:Ppandp}, we get the following useful relations:
\be
\label{eqn:plkandpklandbarpk}
p_{\ell,k} = p_{k,\ell}, \;\; \sum_{k=1}^N p_{\ell,k} = \bar{p}_\ell, \;\; \sum_{\ell=1}^N p_{\ell,k} = \bar{p}_k
\ee

\subsection{Low-Rank Approximation}
We return to our earlier objective of seeking a low-rank factorization for the matrix $(I_{4M^2N^2}-\Fcal^*)^{-1}$. For this purpose, we first introduce the $4M^2\times 4M^2$ Hermitian matrix:
\be
\label{eqn:Fdef}
F \defeq \sum_{k = 1}^{N}\sum_{\ell = 1}^{N} p_{\ell,k} [ \bar{D}_\ell^\T\kron \bar{D}_k + c_{\mu,\ell,k} (H_\ell^\T \kron H_k)]
\ee
where $\bar{D}_k$ is given by \eqref{eqn:Dkmeandef}.

\begin{lemma}[{Spectral radius of $F$}]
\label{lemma:spectralF}
The matrix $F$ in \eqref{eqn:Fdef} is stable if condition \eqref{eqn:meansquarestabilitycond3} is satisfied. Moreover,
\be
\label{eqn:rhoFapprox}
\rho(F) = [1 - \lambda_{\min}(H)]^2 + O(\nu^2) = 1 - O(\nu)
\ee
where
\be
\label{eqn:Hdef}
H \defeq \sum_{k = 1}^{N} \bar{p}_{k} \bar{\mu}_k H_k
\ee
It can be verified that $\| H \| = O(\nu)$ and the $O(\nu^2)$ term in \eqref{eqn:rhoFapprox} is negligible by Assumption \ref{asm:smallstepsizes}.
\end{lemma}
\begin{IEEEproof}
See Appendix \ref{app:spectralF}.
\end{IEEEproof}

\begin{lemma}[{Low-rank approximation}]
\label{lemma:lowrank}
Under Assumptions \ref{asm:smallstepsizes} and \ref{asm:connected}, it holds that
\begin{align}
\label{eqn:lowrankapproxIminusF}
& (I_{4N^2M^2} - \Fcal)^{-1} = (p \one_{N^2}^\T) \kron (I_{4M^2} - F)^{-1} + O(1) \\
\label{eqn:lowrankapproxIminusFdom}
& (p \one_{N^2}^\T) \kron (I_{4M^2} - F)^{-1} = O(\nu^{-1})
\end{align}
Under Assumption \ref{asm:smallstepsizes} where $\nu \ll 1$, the term in \eqref{eqn:lowrankapproxIminusFdom} dominates the $O(1)$ term in \eqref{eqn:lowrankapproxIminusF}. Moreover,
\be
\label{eqn:rhobigFapprox}
\rho(\Fcal) = \rho(F) + O(\nu^{1+1/N^2})
\ee
where the $\rho(F)$ from \eqref{eqn:rhoFapprox} dominates the $O(\nu)$ term in \eqref{eqn:rhobigFapprox}.
\end{lemma}
\begin{IEEEproof}
See Appendix \ref{app:lowrank}.
\end{IEEEproof}

In expression \eqref{eqn:Fdef} we observe that the matrix $F$ is dependent on the first and second-order moments of the random step-sizes, i.e., $\{\bar{\mu}_k\}$ and $\{c_{\mu,\ell,k}\}$, and is also dependent on the first and second order moments of the random combination coefficient matrix, i.e., $\bar{A}$ and $C_A$, through the dependence on the Perron eigenvector $p$. Let us introduce two $4M^2\times 4M^2$ matrices:
\begin{align}
\label{eqn:Rdef}
R & \defeq \sum_{k=1}^{N} p_{k,k} \, ( \bar{\mu}_k^2 + c_{\mu,k,k} ) R_k \\
\label{eqn:Zdef}
Z & \defeq \unvecm\left( (I_{4M^2} - F)^{-1} \vecm(R) \right)
\end{align}
where $R_k$ is given by \eqref{eqn:Rkdef}. Using \eqref{eqn:bigMbound}--\eqref{eqn:bigCMbound2} and \eqref{eqn:lowrankapproxIminusF1} in Appendix \ref{app:lowrank}, we can verify that
\be
\label{eqn:orderofRandZ}
\| R \| = O(\nu^2), \qquad \| Z \| = O(\nu)
\ee 
Then, using Lemma \ref{lemma:lowrank}, we can establish the following useful result about the structure of the steady-state network error covariance matrix.

\begin{theorem}[Network error covariance matrix]
\label{theorem:networkcovariance}
In steady-state, the covariance matrix of the network error $\ubar{\wt{\bm{w}}}_i'$ from the long-term model \eqref{eqn:errorrecursiondefapprox} can be approximated by
\be
\label{eqn:approximatesteadystateCov}
\lim_{i\rightarrow\infty} \E \ubar{\wt{\bm{w}}}_i' \ubar{\wt{\bm{w}}}_i'^* = (\one_N \one_N^\T) \kron Z + O(\nu^{1 + \gamma_o'})
\ee
where $Z$ is from \eqref{eqn:Zdef}, 
\be
\label{eqn:gammaopdef}
\gamma_o' \defeq \frac{1}{2} \min\{ 2, \gamma_v \}
\ee
and the first term on the RHS is dominant.
\end{theorem}
\begin{IEEEproof}
See Appendix \ref{app:networkcovariance}.
\end{IEEEproof}

According to Theorem \ref{theorem:networkcovariance}, the (cross-) covariance matrices of $\{ \ubar{\wt{\bm{w}}}_{k,i}' \}$, which are uniformly expressed by $\E \ubar{\wt{\bm{w}}}_{k,i}' \ubar{\wt{\bm{w}}}_{\ell,i}'^* = \unvecm(z_i^{(\ell,k)})$ for all $k$ and $\ell$ according to \eqref{eqn:zilkdef}, can be approximated by $Z$ in steady-state. However, this result is useful only if $Z$ is a valid complex-Hessian-type matrix.

\begin{definition}[Complex-Hessian-type matrices]
\label{def:complexHessian}
Let $X$ be an $M\times M$ positive semi-definite Hermitian matrix and let $Y$ be an $M\times M$ symmetric matrix. Then, a positive semi-definite block matrix of the form
\be
\label{eqn:complexhessian_structure}
H \defeq \begin{bmatrix}
X   & Y    \\
Y^* & X^\T \\
\end{bmatrix} \ge 0
\ee
will be referred to as a complex-Hessian-type matrix. \hfill \IEEEQED
\end{definition}

The following result explains the reason for introducing this definition.

\begin{lemma}[Complex-Hessian-type covariance matrices]
\label{lemma:cov_complexHessian}
Let $\bm{x}$ denote an $M\times1$ zero-mean complex random vector and let $R_x \defeq \E \bm{x} \bm{x}^* $ and $R_x' \defeq \E \bm{x} \bm{x}^\T $. Then, the covariance matrix of $\ubar{\bm{x}} = \ubar{\mbbT}(\bm{x})$ is given by
\be
\label{eqn:covariance_complexHessian}
\E \ubar{\bm{x}} \ubar{\bm{x}}^* = \begin{bmatrix}
\E \bm{x} \bm{x}^*          & \E \bm{x} \bm{x}^\T        \\
\E (\bm{x}^*)^\T \bm{x}^*   & \E (\bm{x}^*)^\T \bm{x}^\T  \\
\end{bmatrix} = \begin{bmatrix}
R_x     & R_x'   \\
R_x'^*  & R_x^\T \\
\end{bmatrix}   
\ee
and this matrix is a complex-Hessian-type matrix.
\end{lemma}
\begin{IEEEproof}
It follows from comparing \eqref{eqn:covariance_complexHessian} to \eqref{eqn:complexhessian_structure}.
\end{IEEEproof}

By Lemma \ref{lemma:cov_complexHessian}, for \emph{any} zero-mean complex random vector $\bm{x}$, the covariance matrix of $\ubar{\bm{x}} = \ubar{\mbbT}(\bm{x})$ must be a complex-Hessian-type matrix. Therefore, in order to approximate $\{ \unvecm(z_i^{(\ell,k)}) \}$ by $Z$ according to \eqref{eqn:approximatesteadystateCov}, we establish the following useful result for the matrix $Z$.

\begin{lemma}[Properties of $Z$]
\label{lemma:propertiesZ}
The matrix $Z$ in \eqref{eqn:Zdef} is a positive semi-definite complex-Hessian-type matrix. 
\end{lemma}
\begin{IEEEproof}
See Appendix \ref{app:propertiesZ}.
\end{IEEEproof}

Using \eqref{eqn:approximatesteadystateCov} and Lemmas \ref{lemma:cov_complexHessian} and \ref{lemma:propertiesZ}, we arrive at the following result for the covariance and cross-covariance matrices of the steady-state error vectors $\{\ubar{\wt{\bm{w}}}_{k,i}'\}$ from the long-term model \eqref{eqn:errorrecursiondefapprox}.

\begin{corollary}[{Covariance and cross-covariance matrices}]
\label{corollary:crosscov}
The steady-state (cross-) covariance matrices of individual errors $\{ \ubar{\wt{\bm{w}}}_{k,i}' \}$ from the long-term model \eqref{eqn:errorrecursiondefapprox} can be approximated by
\be
\label{eqn:covapprox}
\lim_{i\rightarrow\infty} \E \ubar{\wt{\bm{w}}}_{k,i}'  \ubar{\wt{\bm{w}}}_{\ell,i}'^* = Z + O(\nu^{1 + \gamma_o' })
\ee
for all $k$ and $\ell$, where $Z$ is given by \eqref{eqn:Zdef} and is dominant due to \eqref{eqn:orderofRandZ}, and $\gamma_o'$ is given by \eqref{eqn:gammaopdef}.
\end{corollary}

\begin{IEEEproof}
By Lemma \ref{lemma:propertiesZ}, the $Z$ is a complex-Hessian-type matrix. According to Lemma \ref{lemma:cov_complexHessian}, it is a valid covariance matrix for complex random vectors obtained via the transform $\ubar{\mbbT}(\cdot)$. The approximation \eqref{eqn:covapprox} then follows from Theorem \ref{theorem:networkcovariance}.
\end{IEEEproof}

\subsection{Steady-State MSD}
Using Corollary \ref{corollary:crosscov}, we obtain two useful results about the steady-state MSD for asynchronous diffusion solutions.

\begin{corollary}[{Steady-state MSD}]
\label{corollary:steadystateMSD}
Based on the same assumptions as Theorem \ref{theorem:networkcovariance}, the steady-state MSD (either the network MSD in \eqref{eqn:networkMSDdef} or the individual MSD in \eqref{eqn:individualMSDdef}) can be approximated by 
\begin{align}
\label{eqn:steadyMSDapprox}
\MSD^{\network} & = \frac{1}{4} \Tr( H^{-1} R ) + O(\nu^{1 + \gamma_o }) \\
\label{eqn:steadyMSDapproxind}
\MSD_k & = \frac{1}{4} \Tr( H^{-1} R ) + O(\nu^{1 + \gamma_o })
\end{align}
where $\{ H, R\}$ are given by \eqref{eqn:Hdef} and \eqref{eqn:Rdef}, respectively, $\gamma_o$ is from \eqref{eqn:gammaodef}, and $\Tr( H^{-1} R )$ is of the order of $\nu$.
\end{corollary}
\begin{IEEEproof}
See Appendix \ref{app:steadystateMSD}. 
\end{IEEEproof}

\begin{corollary}[{Clustered solutions}]
\label{corollary:clustered}
The steady-state \emph{relative} MSD between any two agents $k$ and $\ell$, i.e., $\E\,\|\bm{w}_{k,i} - \bm{w}_{\ell,i}\|^2$, is negligible compared to their steady-state \emph{absolute} MSD with respect to $w^o$, i.e.,
\be
\label{eqn:Onu3andOnu2}
\lim_{i\rightarrow\infty} \E\|\bm{w}_{k,i} - \bm{w}_{\ell,i}\|^2 = O(\nu^{1 + \gamma_o' }) \ll \max_k \MSD_k = O(\nu)  
\ee
where $\gamma_o'$ is given by \eqref{eqn:gammaopdef}.
\end{corollary}

\begin{IEEEproof}
First, from Corollary \ref{corollary:crosscov}, we have
\begin{align}
\label{eqn:boundrelativeMSDlongterm}
\lim_{i\rightarrow\infty} \E \| \wt{\ubar{\bm{w}}}_{k,i}' - \wt{\ubar{\bm{w}}}_{\ell,i}' \|^2
& = \lim_{i\rightarrow\infty} \E [ \| \wt{\ubar{\bm{w}}}_{k,i}' \|^2 + \| \wt{\ubar{\bm{w}}}_{\ell,i}' \|^2 - \wt{\ubar{\bm{w}}}_{k,i}'^* \wt{\ubar{\bm{w}}}_{\ell,i}' - \wt{\ubar{\bm{w}}}_{\ell,i}'^* \wt{\ubar{\bm{w}}}_{k,i}' ] \nn \\
& = \Tr( Z ) + \Tr( Z ) - \Tr( Z ) - \Tr( Z ) + O(\nu^{1 + \gamma_o'}) \nn \\
& = O(\nu^{1 + \gamma_o'})
\end{align}
From Theorem \ref{theorem:gapmeansquare} and using \eqref{eqn:boundrelativeMSDlongterm}, we get
\begin{align}
\label{eqn:errorcrossnorm}
\lim_{i\rightarrow\infty} \E\|\wt{\ubar{\bm{w}}}_{k,i} - \wt{\ubar{\bm{w}}}_{\ell,i}\|^2
& = \lim_{i\rightarrow\infty} \E\|\wt{\ubar{\bm{w}}}_{k,i} - \wt{\ubar{\bm{w}}}_{k,i}' + \wt{\ubar{\bm{w}}}_{k,i}' - \wt{\ubar{\bm{w}}}_{\ell,i}' + \wt{\ubar{\bm{w}}}_{\ell,i}' - \wt{\ubar{\bm{w}}}_{\ell,i} \|^2 \nn \\
& \le \lim_{i\rightarrow\infty}   3 \E \left( \| \wt{\ubar{\bm{w}}}_{k,i}  -  \wt{\ubar{\bm{w}}}_{k,i}' \|^2  +  \| \wt{\ubar{\bm{w}}}_{k,i}'  -  \wt{\ubar{\bm{w}}}_{\ell,i}' \|^2  +  \| \wt{\ubar{\bm{w}}}_{\ell,i}'  -  \wt{\ubar{\bm{w}}}_{\ell,i} \|^2 \right) \nn \\
& \le O(\nu^{2}) + O(\nu^{1 + \gamma_o'}) + O(\nu^{2}) \nn \\
& \le O(\nu^{1 + \gamma_o'})
\end{align}
Using \eqref{I-eqn:barw2ubarw} from Part I \cite{Zhao13TSPasync1}, \eqref{eqn:errorcrossnorm}, and Corollary \ref{corollary:steadystateMSD} completes the proof.
\end{IEEEproof}

We illustrated Corollaries \ref{corollary:steadystateMSD} and \ref{corollary:clustered} earlier in Fig.~\ref{fig:normball}. We note that if the cost functions, $\{J_k(w)\}$, are deterministic with known gradient vectors, then performance results in Corollaries \ref{corollary:steadystateMSD} and \ref{corollary:clustered} can still be deduced from our expressions by setting the gradient noise to zero.

\section{Conclusion}
We studied in some detail the MSE performance of \emph{asynchronous} networks with random step-sizes, links, topologies, and combination coefficients. Assuming sufficiently small step-sizes, we showed that at steady-state, the error vector for every individual agent tends to cluster within $O(\nu^{1 + \gamma_o})$ from each other, which means that the MSD performance is essentially uniform across the entire network. The result in Corollary \ref{corollary:steadystateMSD} shows explicitly how the MSD performance of the network is affected by the asynchronous behavior. Quantities that relate to the first and second-order moments of the distribution of the random step-sizes and combination coefficients appear in these expressions. These results can be used to guide strategies for adjusting the combination weights and the rate at which the agents update their solutions and to ensure that the performance (in terms of MSD and rate of convergence) does not degrade below desirable levels.

\appendices

\section{Proof of Lemma \ref{lemma:sizeofdki}}
\label{app:sizeofdki}

Using Lemma \ref{I-lemma:lipschitzHessian} in Part I \cite{Zhao13TSPasync1}, we get from \eqref{eqn:Hkidef}--\eqref{eqn:Hkdef} that
\begin{align}
\label{eqn:Hkapproxerror}
\| H_k - \bm{H}_{k,i-1} \| & = \left\| \int_{0}^{1} [\nabla_{\ubar{w}\ubar{w}^{*}}^2 J_k(\ubar{w}^o) - \nabla_{\ubar{w}\ubar{w}^*}^2 J_k(\ubar{w}^o-t\ubar{\wt{\bm{w}}}_{k,i-1})] dt \right\| \nn \\
& \le \int_{0}^{1} \| \nabla_{\ubar{w}\ubar{w}^{*}}^2 J_k(\ubar{w}^o) - \nabla_{\ubar{w}\ubar{w}^*}^2 J_k(\ubar{w}^o-t\ubar{\wt{\bm{w}}}_{k,i-1}) \| dt \nn \\
& \le \int_{0}^{1} \tau_k' \cdot t \cdot \| \ubar{\wt{\bm{w}}}_{k,i-1} \| \, dt = \frac{\tau_k'}{2} \cdot \| \ubar{\wt{\bm{w}}}_{k,i-1} \|
\end{align}
Using \eqref{eqn:Hkapproxerror}, we get from \eqref{eqn:Deltakidef} that
\begin{align}
\label{eqn:Deltakiapprox}
\| \bm{d}_{k,i} \|^2 & \le [ \bm{\mu}_k(i) ]^2 \cdot \| H_k - \bm{H}_{k, i-1} \|^2 \cdot \| \ubar{\wt{\bm{w}}}_{k,i-1} \|^2 \nn \\
& \le [ \bm{\mu}_k(i) ]^2 \cdot \frac{\tau_k'^2}{4} \cdot \| \ubar{\wt{\bm{w}}}_{k,i-1} \|^4
\end{align}
Taking the expectation of both sides of \eqref{eqn:Deltakiapprox} yields
\be
\label{eqn:Deltakiapprox2}
\E \| \bm{d}_{k,i} \|^2 \le \bar{\mu}_k^{(2)} \frac{\tau_k'^2}{4} \cdot \E \|\ubar{\wt{\bm{w}}}_{k,i-1} \|^4
\ee
From Theorem \ref{I-theorem:4thmoments} of Part I \cite{Zhao13TSPasync1}, it holds for large enough $i$ that
\begin{align}
\label{eqn:bound4thordermoment2}
\E \| \wt{\bm{w}}_{k,i} \|^4 \le 2 b_4^2 \cdot \nu^2
\end{align}
Using the fact from \eqref{I-eqn:boundmumoments2} of Part I \cite{Zhao13TSPasync1} that $\bar{\mu}_k^{(2)} \le \nu^2$ for any $k$, and letting
\be
\label{eqn:taupdef}
\tau' \defeq \max_k \tau_k'
\ee
we obtain from \eqref{eqn:Deltakiapprox2} and \eqref{eqn:bound4thordermoment2} that
\be
\label{eqn:boundnorm2phiki1}
\E \| \bm{d}_{k,i} \|^2 \le 2 \tau'^2 b_4^2 \cdot \nu^4 = O(\nu^4)
\ee
where a factor of 4 appeared due to the conversion $\ubar{\mbbT}(\cdot)$ from \eqref{I-eqn:ubardef} of Part I \cite{Zhao13TSPasync1}. Then,
\begin{align}
\label{eqn:boundnorm2phiki}
\E \| \bm{\Acal}_i^\T \bm{d}_i \|^2 & = \sum_{k=1}^N \E \left\| \sum_{\ell \in \bm{\Ncal}_{k,i}} \bm{a}_{\ell k}(i) \bm{d}_{\ell,i} \right\|^2 \nn \\
& \stackrel{(a)}{\le} \sum_{k=1}^N \E \left[ \sum_{\ell \in \bm{\Ncal}_{k,i}} \bm{a}_{\ell k}(i) \| \bm{d}_{\ell,i} \|^2 \right] \nn \\
& = \sum_{k=1}^N \sum_{\ell \in \Ncal_k} \bar{a}_{\ell k} \E \| \bm{d}_{\ell,i} \|^2 \nn \\
& \le N \cdot \max_\ell \E \| \bm{d}_{\ell,i} \|^2 \nn \\
& \stackrel{(b)}{\le} 2 N \tau'^2 b_4^2 \cdot \nu^4
\end{align}
where step (a) is by using Jensen's inequality; and step (b) is by using \eqref{eqn:boundnorm2phiki1}.

\section{Proof of Theorem \ref{theorem:gapmeansquare}}
\label{app:boundgapmeansquare}
We rewrite the original error recursion \eqref{eqn:errorrecursiondefdelta} and the long-term model \eqref{eqn:errorrecursiondefapprox} respectively as follows:
\begin{align}
\label{eqn:atcasync1errorubar1org1}
\ubar{\wt{\bm{\psi}}}_{k,i} & = [I_{2M} - \bm{\mu}_k(i) H_k] \ubar{\wt{\bm{w}}}_{k,i-1} + \bm{\mu}_k(i) \ubar{\bm{v}}_{k,i}(\bm{w}_{k,i-1}) + \bm{d}_{k,i} \\
\label{eqn:atcasync2errorubar1org1}
\ubar{\wt{\bm{w}}}_{k,i} & = \sum_{\ell\in\bm{\Ncal}_{k,i}} \bm{a}_{\ell k}(i)\,\ubar{\wt{\bm{\psi}}}_{\ell,i}
\end{align}
and
\begin{align}
\label{eqn:atcasync1errorubar1refnew1}
\ubar{\wt{\bm{\psi}}}_{k,i}' & = [I_{2M} - \bm{\mu}_k(i) H_k] \ubar{\wt{\bm{w}}}_{k,i-1}' + \bm{\mu}_k(i) \ubar{\bm{v}}_{k,i}(\bm{w}_{k,i-1}) \\
\label{eqn:atcasync2errorubar1refnew1}
\ubar{\wt{\bm{w}}}_{k,i}' & = \sum_{\ell\in\bm{\Ncal}_{k,i}} \bm{a}_{\ell k}(i)\,\ubar{\wt{\bm{\psi}}}_{\ell,i}'
\end{align}
with the prime notation for quantities associated with the long-term  model \eqref{eqn:errorrecursiondefapprox}. From \eqref{eqn:atcasync2errorubar1org1} and \eqref{eqn:atcasync2errorubar1refnew1}, and using Jensen's inequality, the squared 2-norm of the difference between the two models is given by
\begin{align}
\label{eqn:meansquaregapw1}
\| \ubar{\wt{\bm{w}}}_{k,i} - \ubar{\wt{\bm{w}}}_{k,i}' \|^2 \le \sum_{\ell\in\bm{\Ncal}_{k,i}} \bm{a}_{\ell k}(i)\, \| \ubar{\wt{\bm{\psi}}}_{\ell,i} - \ubar{\wt{\bm{\psi}}}_{\ell,i}'  \|^2
\end{align}
Taking the expectation of both sides yields
\begin{align}
\E \| \ubar{\wt{\bm{w}}}_{k,i} - \ubar{\wt{\bm{w}}}_{k,i}' \|^2 & \le \sum_{\ell\in\Ncal_k} \bar{a}_{\ell k}\, \E \| \ubar{\wt{\bm{\psi}}}_{\ell,i} - \ubar{\wt{\bm{\psi}}}_{\ell,i}' \|^2 \nn \\
& \le \max_{\ell} \, \E \| \ubar{\wt{\bm{\psi}}}_{\ell,i} - \ubar{\wt{\bm{\psi}}}_{\ell,i}' \|^2
\end{align}
for all $k$. Then,
\be
\label{eqn:maxwgap2andmaxpsigap2}
\max_k \E \| \ubar{\wt{\bm{w}}}_{k,i} - \ubar{\wt{\bm{w}}}_{k,i}' \|^2  \le \max_{k} \, \E \| \ubar{\wt{\bm{\psi}}}_{k,i} - \ubar{\wt{\bm{\psi}}}_{k,i}' \|^2
\ee
From \eqref{eqn:atcasync1errorubar1org1} and \eqref{eqn:atcasync1errorubar1refnew1}, we have
\be
\ubar{\wt{\bm{\psi}}}_{k,i} - \ubar{\wt{\bm{\psi}}}_{k,i}' = [I_{2M} - \bm{\mu}_k(i) H_k] ( \ubar{\wt{\bm{w}}}_{k,i-1} - \ubar{\wt{\bm{w}}}_{k,i-1}' ) + \bm{d}_{k,i} 
\ee
Taking the expected squared 2-norm of both sides, we have
\begin{align}
\label{eqn:meansquaregappsi}
\E \| \ubar{\wt{\bm{\psi}}}_{k,i} - \ubar{\wt{\bm{\psi}}}_{k,i}' \|^2 & \le \E \| [I_{2M} - \bm{\mu}_k(i) H_k] ( \ubar{\wt{\bm{w}}}_{k,i-1} - \ubar{\wt{\bm{w}}}_{k,i-1}' ) + \bm{d}_{k,i} \|^2 \nn \\
& = \E \left\| (1 - t) \frac{I_{2M} - \bm{\mu}_k(i) H_k}{1 - t} ( \ubar{\wt{\bm{w}}}_{k,i-1} - \ubar{\wt{\bm{w}}}_{k,i-1}' ) + t \frac{\bm{d}_{k,i}}{t} \right\|^2 \nn \\
& \le (1 - t)^{-1} \E \| I_{2M} - \bm{\mu}_k(i) H_k \|^2 \cdot \E \| \ubar{\wt{\bm{w}}}_{k,i-1} - \ubar{\wt{\bm{w}}}_{k,i-1}' \|^2 + t^{-1} \E \| \bm{d}_{k,i} \|^2
\end{align}
for any $0 < t < 1$, where we used Jensen's inequality in the second inequality. By condition \eqref{eqn:boundstepsize4thorder}, it can be verified that
\begin{align}
\label{eqn:boundIminusmukiHknorm2}
\E \| I_{2M} - \bm{\mu}_k(i) H_k \|^2 & \le 1 - 2 \bar{\mu}_k^{(1)} \lambda_{k,\min} + \bar{\mu}_k^{(2)} \lambda_{k,\max}^2 \nn \\
& \stackrel{(a)}{\le} 1 - \bar{\mu}_k^{(1)} \lambda_{k,\min} \nn \\
& \le \left( 1 - \frac{1}{2} \bar{\mu}_k^{(1)} \lambda_{k,\min} \right)^2  < 1
\end{align}
where step (a) is from \eqref{I-eqn:stability4thorder2new} of Part I \cite{Zhao13TSPasync1}. Substituting $t = \frac{1}{2} \bar{\mu}_k^{(1)} \lambda_{k,\min} < 1$ and \eqref{eqn:boundIminusmukiHknorm2} into \eqref{eqn:meansquaregappsi} yields
\begin{align}
\label{eqn:meansquaregappsi0}
\E \| \ubar{\wt{\bm{\psi}}}_{k,i} - \ubar{\wt{\bm{\psi}}}_{k,i}' \|^2 & \le \left( 1 - \frac{ 1}{2} \bar{\mu}_k^{(1)} \lambda_{k,\min} \right) \E \| \ubar{\wt{\bm{w}}}_{k,i-1} - \ubar{\wt{\bm{w}}}_{k,i-1}' \|^2 + \frac{2}{\lambda_{k,\min} \bar{\mu}_k^{(1)}} \E \| \bm{d}_{k,i} \|^2
\end{align}
Using \eqref{eqn:Deltakiapprox2}, the second term on the RHS of \eqref{eqn:meansquaregappsi0} can be bounded for large enough $i$ by
\begin{align}
\label{eqn:boundnorm2phiki1new}
\frac{2}{\lambda_{k,\min} \bar{\mu}_k^{(1)}} \cdot \E \| \bm{d}_{k,i} \|^2 & \le \frac{2}{\lambda_{k,\min} \bar{\mu}_k^{(1)}} \cdot \bar{\mu}_k^{(2)} \frac{\tau_k'^2}{4} \cdot \E \|\ubar{\wt{\bm{w}}}_{k,i-1} \|^4 \nn \\
& = \frac{\tau_k'^2}{2 \lambda_{k,\min}} \cdot \frac{\bar{\mu}_k^{(2)}}{\bar{\mu}_k^{(1)}} \cdot \E \|\ubar{\wt{\bm{w}}}_{k,i-1} \|^4 \nn \\
& \le \frac{\tau_k'^2 \nu}{2 \lambda_{k,\min}} \cdot \E \|\ubar{\wt{\bm{w}}}_{k,i-1} \|^4
\end{align}
where we used the fact from \eqref{I-eqn:oldnulessnewnu} of Part I \cite{Zhao13TSPasync1} that $\nu \ge {\bar{\mu}_k^{(2)}}/{\bar{\mu}_k^{(1)}}$. Substituting \eqref{eqn:boundnorm2phiki1new} into \eqref{eqn:meansquaregappsi0} yields
\begin{align}
\label{eqn:meansquaregappsi1}
\E \| \ubar{\wt{\bm{\psi}}}_{k,i} - \ubar{\wt{\bm{\psi}}}_{k,i}' \|^2 & \le \left( 1 - \frac{1}{2} \bar{\mu}_k^{(1)} \lambda_{k,\min} \right) \E \| \ubar{\wt{\bm{w}}}_{k,i-1} - \ubar{\wt{\bm{w}}}_{k,i-1}' \|^2 + \frac{\tau_k'^2 \nu }{2 \lambda_{k,\min}} \E \|\ubar{\wt{\bm{w}}}_{k,i-1} \|^4 
\end{align}
Therefore,
\begin{align}
\label{eqn:meansquaregappsi3}
\max_k \E \| \ubar{\wt{\bm{\psi}}}_{k,i} - \ubar{\wt{\bm{\psi}}}_{k,i}' \|^2 & \le \max_k \left( 1 - \frac{1}{2} \bar{\mu}_k^{(1)} \lambda_{k,\min} \right) \max_k \E \| \ubar{\wt{\bm{w}}}_{k,i-1} - \ubar{\wt{\bm{w}}}_{k,i-1}' \|^2 \nn \\
& \qquad + \max_k \left[ \frac{\tau_k'^2 \nu }{2 \lambda_{k,\min}} \cdot \E \|\ubar{\wt{\bm{w}}}_{k,i-1} \|^4 \right]
\end{align}
Substituting \eqref{eqn:meansquaregappsi3} into \eqref{eqn:maxwgap2andmaxpsigap2} yields
\begin{align}
\label{eqn:meansquaregapw}
\max_k \E \| \ubar{\wt{\bm{w}}}_{k,i} - \ubar{\wt{\bm{w}}}_{k,i}' \|^2 & \le \gamma \cdot \max_k \E \| \ubar{\wt{\bm{w}}}_{k,i-1} - \ubar{\wt{\bm{w}}}_{k,i-1}' \|^2 + \frac{\tau'^2 \nu}{2 \min_k \lambda_{k,\min}} \cdot \max_k \E \|\ubar{\wt{\bm{w}}}_{k,i-1} \|^4 
\end{align}
where
\be
\label{eqn:gammadef}
\gamma \defeq \max_k \left( 1 - \frac{1}{2} \bar{\mu}_k^{(1)} \lambda_{k,\min} \right) = 1 - \frac{1}{2} \min_k \{ \bar{\mu}_k^{(1)} \cdot \lambda_{k,\min} \}
\ee
When condition \eqref{eqn:boundstepsize4thorder} holds, it can be verified by using \eqref{I-eqn:stability4thorder1new} from Part I \cite{Zhao13TSPasync1} that $| \gamma | < 1$. Then, we get from \eqref{eqn:meansquaregapw} that
\begin{align}
\limsup_{i\rightarrow\infty} \left[ \max_k \E \| \ubar{\wt{\bm{w}}}_{k,i} - \ubar{\wt{\bm{w}}}_{k,i}' \|^2 \right] & \le \frac{\tau'^2 \nu}{(1-\gamma) \cdot 2 \min_k \lambda_{k,\min} } \cdot \limsup_{i\rightarrow\infty} \left[ \max_k \E \|\ubar{\wt{\bm{w}}}_{k,i-1} \|^4 \right] \nn \\
& \le \frac{4 \tau'^2 b_4^2 \nu^3}{\min_k \bar{\mu}_k^{(1)} \cdot \min_k \lambda_{k,\min}^2 } \le O(\nu^2)
\end{align}
where we used Theorem \ref{I-theorem:4thmoments} from Part I \cite{Zhao13TSPasync1} and the fact from \eqref{I-eqn:boundmumoments2} of Part I \cite{Zhao13TSPasync1} that $\bar{\mu}_k^{(1)} = O(\nu)$.

\section{Block Operations}
\label{app:blockvecandkronecker}
Consider a block matrix $\Xcal$ of size $NM \times NM$ and partition it into $N \times N$ blocks where $X_{k \ell}$ denotes its $(k,\ell)$-th sub-matrix of size $M\times M$. The block vectorization of $\Xcal$ with block size $M\times M$ is defined as follows \cite{Koning91LAA}:
\begin{align}
\label{eqn:bvecdef}
\bvec(\Xcal) & \defeq \col\{ \vecm(X_{11}), \vecm(X_{21}), \dots, \vecm(X_{N1}), \dots, \vecm(X_{1N}), \vecm(X_{2N}), \dots, \vecm(X_{NN}) \}
\end{align}
Let $\Ycal$ denote another block matrix of size $NM \times NM$ and let $Y_{k \ell}$ denote its $(k,\ell)$-th sub-matrix of size $M\times M$. Then, the block Kronecker product of $\Xcal$ and $\Ycal$ with block size $M \times M$ is defined by \cite{Koning91LAA}:
\be 
\label{eqn:bkrondef1}
\Xcal \bkron \Ycal \defeq \begin{bmatrix}
Z_{11} & Z_{12} & \dots & Z_{1N} \\
Z_{21} & Z_{22} & \dots & Z_{2N} \\
\vdots & \vdots & \ddots & \vdots \\
Z_{N1} & Z_{N2} & \dots & Z_{NN} \\
\end{bmatrix}
\ee
where
\be
\label{eqn:bkrondef2}
Z_{k \ell} \defeq \begin{bmatrix}
X_{k \ell} \kron Y_{11}   &   X_{k \ell} \kron Y_{12}   &   \dots   &   X_{k \ell} \kron Y_{1N} \\
X_{k \ell} \kron Y_{21}   &   X_{k \ell} \kron Y_{22}   &  \dots   &   X_{k \ell} \kron Y_{2N} \\
\vdots   &   \vdots   &   \ddots   &   \vdots \\
X_{k \ell} \kron Y_{N1}   &   X_{k \ell} \kron Y_{N2}   &  \dots   &   X_{k \ell} \kron Y_{NN} \\
\end{bmatrix} 
\ee
For any matrices $\{X,Y,A,B\}$ of compatible dimensions and with blocks of size $M \times M$, it holds that
\be
\label{eqn:bkronandkron}
(X \kron A) \bkron (Y \kron B) = (X \kron Y) \kron (A \kron B)
\ee
where $\kron$ denotes the traditional Kronecker product operation. Other useful properties for the $\bkron$ operation can be found in \cite[pp.~176-179]{Koning91LAA} and are listed here for ease of reference:
\begin{align}
\label{eqn:bvecproduct} 
\bvec( \Acal \Bcal \Ccal) & = (\Ccal^\T \bkron \Acal) \cdot \bvec(\Bcal) \\
\label{eqn:bvecproduct_vec} 
\bvec(x y^{\T}) & = y \bkron x \\
\label{eqn:bvectrace} 
\Tr(\Acal \Bcal) & = [\bvec(\Acal^\T)]^\T \cdot \bvec(\Bcal) = [\bvec(\Acal^*)]^* \cdot \bvec(\Bcal) \\
\label{eqn:bvecmultiply}
(\Acal \Ccal) \bkron (\Bcal \Dcal) & = (\Acal \bkron \Bcal)(\Ccal \bkron \Dcal) \\
\label{eqn:bvecadd}
(\Acal + \Bcal) \bkron (\Ccal + \Dcal) & = \Acal \bkron \Ccal + \Bcal \bkron \Ccal + \Acal \bkron \Dcal + \Bcal \bkron \Dcal \\
\label{eqn:bvecHerm}
(\Acal \bkron \Bcal)^* & = \Acal^* \bkron \Bcal^* \\
\label{eqn:bvecSymm}
(\Acal \bkron \Bcal)^\T & = \Acal^\T \bkron \Bcal^\T
\end{align}
for any block matrices $\{\Acal, \Bcal, \Ccal, \Dcal\}$ and any block vectors $\{x, y\}$ with appropriate sizes.

\section{Proof of Theorem \ref{theorem:errorcovariancerecursion}}
\label{app:errorcovariancerecursion}
From the long-term model \eqref{eqn:errorrecursiondefapprox_new}, we obtain that
\begin{align}
\label{eqn:errorcov_app1}
\E(\ubar{\wt{\bm{w}}}_i' \ubar{\wt{\bm{w}}}_i'^* | \F_{i-1} ) & = \E(\bm{\Bcal}_i \ubar{\wt{\bm{w}}}_{i-1}' \ubar{\wt{\bm{w}}}_{i-1}'^* \bm{\Bcal}_i^* | \mbbF_{i-1} ) + \E(\ubar{\bm{s}}_i\ubar{\bm{s}}_i^* | \F_{i-1} )
\end{align}
where the cross terms that involve $\ubar{\bm{s}}_i$ disappear because $\E(\bm{\Bcal}_i \ubar{\wt{\bm{w}}}_{i-1}\ubar{\bm{s}}_i^* | \F_{i-1} ) = 0$ by \eqref{eqn:bigsidef} and \eqref{eqn:bigsmeandef}. Performing the block vectorization of block size $2M$ for both sides of \eqref{eqn:errorcov_app1}, and using \eqref{eqn:bvecproduct} and \eqref{eqn:bvecproduct_vec} yield
\begin{align}
\label{eqn:errorcov_app4}
\E[(\ubar{\wt{\bm{w}}}_i'^*)^\T \bkron \ubar{\wt{\bm{w}}}_i' | \F_{i-1} ] & = \E[ (\bm{\Bcal}_i^*)^\T \bkron \bm{\Bcal}_i]  [(\ubar{\wt{\bm{w}}}_{i-1}'^*)^\T \bkron \ubar{\wt{\bm{w}}}_{i-1}'] + \E[(\ubar{\bm{s}}_i^*)^\T \bkron \ubar{\bm{s}}_i | \F_{i-1} ]
\end{align}
Using \eqref{eqn:bigsidef}, \eqref{eqn:bvecmultiply}, and \eqref{eqn:bvecSymm}, the second term on the RHS of \eqref{eqn:errorcov_app4} can be expressed as
\begin{align}
\label{eqn:Ess}
\E[ (\ubar{\bm{s}}_i^*)^\T \bkron \ubar{\bm{s}}_i | \F_{i-1} ] & = (\bar{\Acal} \bkron \bar{\Acal} + \Ccal_A)^\T (\bar{\Mcal} \bkron \bar{\Mcal} + \Ccal_M) r_i( \bm{w}_{i-1} )
\end{align}
Substituting \eqref{eqn:bigFdef} and \eqref{eqn:Ess} into \eqref{eqn:errorcov_app4}, taking the expectation with respect to $\F_{i-1}$, and then using \eqref{eqn:bigyidef}, we arrive at the desired recursion \eqref{eqn:meansquareerrorrecursion}, namely,
\be
\label{eqn:errorrecursion2ndorderveccondexp3}
\E\,[(\ubar{\wt{\bm{w}}}_i'^*)^\T \bkron \ubar{\wt{\bm{w}}}_i'] = \Fcal^* \E\,[(\ubar{\wt{\bm{w}}}_{i-1}'^*)^\T \bkron \ubar{\wt{\bm{w}}}_{i-1}'] + y_i
\ee
From \eqref{eqn:bigFdef}, we know that $\Gcal$ in \eqref{eqn:bigGdef} is a factor of $\Fcal$. Hence, we use the following result to examine the stability of $\Fcal$.

\begin{lemma}[Properties of $\Gcal$]
\label{lemma:stablebigG}
The matrix $\Gcal$ in \eqref{eqn:bigGdef} satisfies the following properties:
\begin{enumerate}
\item Block diagonal and Hermitian matrix: it holds that
\begin{align}
\label{eqn:bigGpropertyBD1}
\Gcal & = \diag\{ G_{1}, G_2, \dots, G_N \} \\
\label{eqn:bigGpropertyBD2}
G_{\ell} & = \diag\{ G_{\ell,1}, G_{\ell,2}, \dots, G_{\ell,N}\}
\end{align}
where $G_{\ell}$ denotes the $\ell$-th block on the diagonal of $\Gcal$ with block size $4M^2N\times 4M^2N$ and $G_{\ell,k}$ denotes the $k$-th block on the diagonal of $G_\ell$ with block size $4M^2\times 4M^2$. The block $G_{\ell,k}$ is Hermitian and is given by
\be
\label{eqn:Glkdef}
G_{\ell,k} \defeq \bar{D}_\ell^\T \kron \bar{D}_k + c_{\mu,\ell,k} (H_\ell^\T \kron H_k)
\ee
where $\bar{D}_k$ is given by \eqref{eqn:Dkmeandef}.

\item Norms and spectral radius: it can be verified that
\be
\label{eqn:bigGpropertyNorms}
\rho(\Gcal) = \max_{k,m}\{ (1-\bar{\mu}_k\lambda_{k,m})^2 + c_{\mu,k,k} \lambda_{k,m}^2 \}
\ee
where $\lambda_{k,m}$ denotes the $m$-th eigenvalue of $H_k$, $m = 1,2,\dots, 2M$.

\item Stability: if condition \eqref{eqn:meansquarestabilitycond3} holds, then
\be
\label{eqn:bigGpropertystable}
\rho(\Gcal) < 1
\ee
\end{enumerate}
\end{lemma}
\begin{IEEEproof}
See Appendix \ref{app:proofstablebigG}.
\end{IEEEproof}

Using the fact that $\Gcal$ is block diagonal and Hermitian, and that $\bar{A} \kron \bar{A} + C_A$ is block left-stochastic, result (153) from \cite[App. A]{Zhao12TSP} implies that
\be
\label{eqn:rhobigFandbigG}
\rho(\Fcal) \le \rho(\Gcal)
\ee
By \eqref{eqn:bigGpropertystable} and \eqref{eqn:rhobigFandbigG}, we conclude that $\Fcal$ is stable if condition \eqref{eqn:meansquarestabilitycond3} holds.

\section{Proof of Lemma \ref{lemma:stablebigG}}
\label{app:proofstablebigG}
The first property relating to the block diagonal and Hermitian structure of \eqref{eqn:bigGpropertyBD1}--\eqref{eqn:Glkdef} is established by using the definition of $\bkron$ and \eqref{eqn:bigGdef}. Because the matrix $\bm{\Dcal}_i$ is block diagonal with block size $2M\times 2M$, the block Kronecker product:
\be
\label{eqn:bigGidef}
\bm{\Gcal}_i \defeq \bm{\Dcal}_i^\T \bkron \bm{\Dcal}_i
\ee
is block diagonal with block size $4M^2N\times 4M^2N$ and each block is itself block diagonal with block size $4M^2\times 4M^2$. Let us denote the $\ell$-th block on the diagonal of $\bm{\Gcal}_i$ with block size $4M^2N \times 4M^2N$ by 
\be
\label{eqn:bigGlidef}
\bm{G}_{\ell,i} \defeq \bm{D}_{\ell,i}^\T \kron \bm{\Dcal}_i
\ee
and the $k$-th block on the diagonal of $\bm{G}_{\ell,i}$ with block size $4M^2 \times 4M^2$ by
\be
\label{eqn:bigGlkidef}
\bm{G}_{\ell,k,i} \defeq \bm{D}_{\ell,i}^\T \kron \bm{D}_{k,i}
\ee
where we used the fact that $\bm{D}_{\ell,i}$ is Hermitian. Then, we have
\begin{align}
\label{eqn:bigGiblock}
\bm{\Gcal}_i & = \diag\{ \bm{G}_{1,i}, \bm{G}_{2,i}, \dots, \bm{G}_{N,i} \} \\
\label{eqn:Gliblock}
\bm{G}_{\ell,i} & = \diag\{ \bm{G}_{\ell, 1,i}, \bm{G}_{\ell, 2,i}, \dots, \bm{G}_{\ell, N,i} \}
\end{align}
Using \eqref{eqn:bigGdef} and taking the expectation of both sides of \eqref{eqn:bigGiblock} and \eqref{eqn:Gliblock}, we get \eqref{eqn:bigGpropertyBD1} and \eqref{eqn:bigGpropertyBD2}
by identifying:
\be
\label{eqn:GlkandGlki}
G_\ell = \E[\bm{G}_{\ell,i}], \qquad G_{\ell,k} = \E[\bm{G}_{\ell,k,i}]
\ee
Equation \eqref{eqn:Glkdef} follows from \eqref{eqn:GlkandGlki}, \eqref{eqn:bigGlkidef}, \eqref{eqn:Dkidef}, and \eqref{eqn:Dkmeandef}. Since the matrices $\{G_{\ell,k}\}$ are all Hermitian, by \eqref{eqn:bigGpropertyBD1} and \eqref{eqn:bigGpropertyBD2}, the matrix $\Gcal$ is also Hermitian. 

The second property in \eqref{eqn:bigGpropertyNorms} is established by using the block diagonal and Hermitian properties of $\Gcal$ to readily conclude that $\rho(\Gcal) = \max_{\ell,k}\rho(G_{\ell,k})$. Furthermore, by \eqref{eqn:Glkdef}, the eigenvalues of $G_{\ell,k}$ are given by
\be
\label{eqn:Glkdefeig}
\lambda_{m,n}( G_{\ell,k} ) = (1-\bar{\mu}_\ell\lambda_{\ell,n})(1-\bar{\mu}_k \lambda_{k,m}) + c_{\mu,\ell,k} \lambda_{\ell, n}\lambda_{k,m}
\ee
for any $\ell,k = 1,2,\dots, N$, where $\lambda_{k,m}$ denotes the $m$-th eigenvalue of $H_k$ and $m,n = 1,2,\dots,2M$. It is straightforward to verify that
\be
\label{eqn:oneminusmucrosscov}
\lambda_{m,n}( G_{\ell,k} ) = \E [(1 - \bm{\mu}_\ell(i) \lambda_{\ell,n}) ( 1 -  \bm{\mu}_k(i) \lambda_{k,m} )]
\ee
Using Cauchy-Schwarz inequality, we get
\begin{align}
\label{eqn:eiginequ}
|\E [(1 - \bm{\mu}_\ell(i) \lambda_{\ell,n} ) ( 1-  \bm{\mu}_k(i) \lambda_{k,m} ) ]| & \le \sqrt{\E [(1 - \bm{\mu}_\ell(i) \lambda_{\ell,n} )^2 ] \cdot \E [ (1-  \bm{\mu}_k(i) \lambda_{k,m})^2 ] } \nn \\
& \le \max_{k,m}\{ \E[ (1-  \bm{\mu}_k(i) \lambda_{k,m} )^2 ] \}  \nn \\
& = \max_{k,m}\{ (1-\bar{\mu}_k\lambda_{k,m})^2 + c_{\mu,k,k} \lambda_{k,m}^2 \} 
\end{align}
where the first inequality becomes equality when $\ell = k$ and $n = m$. From \eqref{eqn:Glkdefeig}--\eqref{eqn:eiginequ} we get
\be
\label{eqn:eigbound}
|\lambda_{m,n}( G_{\ell,k} )| \le \max_{k,m}\{ (1-\bar{\mu}_k\lambda_{k,m})^2 + c_{\mu,k,k} \lambda_{k,m}^2 \} 
\ee
for any $\ell$, $k$, $m$, and $n$. Since the above inequality applies to \emph{all} eigenvalues of $G_{\ell,k}$, and since $G_{\ell,k}$ is Hermitian, we get
\be
\label{eqn:Glk2normbbound}
\rho(G_{\ell,k}) \le \max_{k,m}\{ (1-\bar{\mu}_k\lambda_{k,m})^2 + c_{\mu,k,k} \lambda_{k,m}^2 \} 
\ee
Furthermore, from \eqref{eqn:Glkdefeig} we know that 
\be
\lambda_{m,m}( G_{k,k} ) = (1-\bar{\mu}_k\lambda_{k,m})^2 + c_{\mu,k,k} \lambda_{k,m}^2
\ee
so that equality in \eqref{eqn:Glk2normbbound} is achievable for some $k$ and $m$.

For the third property in \eqref{eqn:bigGpropertystable}, we introduce the quadratic function 
\be
\label{eqn:funcfdef}
f(x) \defeq (1-\bar{\mu}_k x)^2 + c_{\mu,k,k} x^2
\ee
with $x\in[\lambda_{k,\min}, \lambda_{k,\max}]$. It is easy to verify that $f(x)$ achieves its maximum value at either one of its boundaries:
\begin{align}
\label{eqn:fxmax}
f(x) & \le \max\{ f(\lambda_{k,\min}), f(\lambda_{k,\max}) \} \nn \\
& \le 1 - 2\bar{\mu}_k\lambda_{k,\min} + ( \bar{\mu}_k^2 + c_{\mu,k,k}) \lambda_{k,\max}^2
\end{align}
From Assumption \ref{I-asm:boundedHessian} in Part I \cite{Zhao13TSPasync1} we have $\lambda_{k,\min} \le \lambda_{k,m} \le \lambda_{k,\max}$ for any $k$ and $m$. We then deduce from \eqref{eqn:fxmax} that
\be
\label{eqn:flambdakmbound}
f(\lambda_{k,m}) \le 1 - 2\bar{\mu}_k\lambda_{k,\min} + ( \bar{\mu}_k^2 + c_{\mu,k,k}) \lambda_{k,\max}^2
\ee
for any $k$ and $m$. Using \eqref{eqn:bigGpropertyNorms}, \eqref{eqn:funcfdef}, and \eqref{eqn:flambdakmbound}, we get
\begin{align}
\label{eqn:bigFstable2}
\rho(\Gcal) & = \max_{k,m} \{ f(\lambda_{k,m}) \} \nn \\
& \le \max_{k}\{ f(\lambda_{k,\min}), f(\lambda_{k,\max}) \} \nn \\
& < \max_{k}\{ f(\lambda_{k,\min}), f(\lambda_{k,\max}) \} + \alpha ( \bar{\mu}_k^2 + c_{\mu,k,k})   
\end{align}
where $\alpha > 0$ by Assumption \ref{asm:gradientnoisestronger}. 
When condition \eqref{eqn:meansquarestabilitycond3} holds, using \eqref{I-eqn:meansquarestabilitycond2} from Part I \cite{Zhao13TSPasync1}, we have
\be
\label{eqn:fminmaxinequality1}
\max_{k}\{ 1 - 2\bar{\mu}_k\lambda_{k,\min} + ( \bar{\mu}_k^2 + c_{\mu,k,k}) ( \lambda_{k,\max}^2 + \alpha) \} < 1
\ee
Therefore, by \eqref{eqn:bigFstable2} and \eqref{eqn:fminmaxinequality1}, if condition \eqref{eqn:meansquarestabilitycond3} holds, then $\rho(\Gcal) < 1$, which completes the proof.

\section{Proof of Lemma \ref{lemma:primitiveA}}
\label{app:primitive}
From Lemma \ref{I-lemma:leftstochastic} in Part I \cite{Zhao13TSPasync1} we know that the matrices $\bar{A} \kron \bar{A} + C_A$ and $\bar{A}$ are both left-stochastic. To establish the desired result, we only need to show that the matrix $\bar{A} \kron \bar{A}$ is primitive if $\bar{A} \kron \bar{A} + C_A$ is primitive. This is because if $\bar{A} \kron \bar{A}$ is primitive, then for some finite  positive integer $j > 0$, the matrix $(\bar{A} \kron \bar{A})^j$ has strictly positive entries. Since $(\bar{A} \kron \bar{A})^j = \bar{A}^j \kron \bar{A}^j$ and $\bar{A}$ has nonnegative entries, $\bar{A}^j$ must have strictly positive entries. Therefore, $\bar{A}$ is primitive.

In order to prove that the matrix $\bar{A} \kron \bar{A}$ is primitive if $\bar{A} \kron \bar{A} + C_A$ is primitive, we first introduce the following concept. 

\begin{definition}[Comparing sparsity] 
\label{def:sparsity}
For any two matrices $\{A, B\}$ with nonnegative entries and of the same size, the matrix $A$ is called \emph{sparser} than $B$, or, equivalently, $B$ is called \emph{denser} than $A$, if, and only if, $[B]_{\ell k} > 0$ whenever $[A]_{\ell k} > 0$ for any $k$ and $\ell$. \hfill \IEEEQED
\end{definition}

It is straightforward to verify the following three useful properties related to Definition \ref{def:sparsity}.

\begin{lemma}[Denser product]
\label{lemma:denserproduct}
For any $M \times N$ matrices $\{A, B\}$ and any $N \times P$ matrices $\{C, D\}$ all with nonnegative entries, if $B$ is denser than $A$ and $D$ is denser than $C$, then $BD$ is denser than $AC$. \hfill \IEEEQED
\end{lemma}

\begin{lemma}[Denser Kronecker product]
\label{lemma:denserKronecker}
For any $M \times N$ matrices $\{A, B\}$ and any $P \times Q$ matrices $\{C, D\}$ all with nonnegative entries, if $B$ is denser than $A$ and $D$ is denser than $C$, then $B \kron D$ is denser than $A \kron C$. \hfill \IEEEQED
\end{lemma}

\begin{lemma}[Sum is not denser]
\label{lemma:densersum}
For any set of $M \times N$ matrices $\{A_i\}$ with nonnegative entries, where $i \in \Ical$ and $\Ical$ is an index set (which can be uncountable), if there exists an $M \times N$ matrix $B$ with nonnegative entries such that $B$ is denser than every $A_i$, $i\in\Ical$, and assuming that the sum $S \defeq \sum_{i\in\Ical}A_i$ exists, then $B$ is also denser than $S$. \hfill \IEEEQED
\end{lemma}

Now, from Lemma \ref{I-lemma:neighborhoods} in Part I \cite{Zhao13TSPasync1}, we know that $\bar{A}$ is denser than any realization of $\bm{A}_i$, say, $\bm{A}_i(\omega)$ where $\omega\in\Omega$ and $\Omega$ is the sample space of $\bm{A}_i$. Using Lemma \ref{lemma:denserKronecker}, we get that $\bar{A} \kron \bar{A}$ is denser than any $\bm{A}_i(\omega) \kron \bm{A}_i(\omega)$. Using Lemma \ref{lemma:densersum} and the fact that the probability measures only take nonnegative values, we  get that $\bar{A} \kron \bar{A}$ is denser than $\bar{A} \kron \bar{A} + C_A = \E\,[\bm{A}_i \kron \bm{A}_i]$. If $\bar{A} \kron \bar{A} + C_A$ is primitive, then there exists a finite positive integer $j>0$ such that $(\bar{A} \kron \bar{A} + C_A)^j$ has strictly positive entries. Using Lemma \ref{lemma:denserproduct}, we know that $(\bar{A} \kron \bar{A})^j$ must be denser than $(\bar{A} \kron \bar{A} + C_A)^j$. Therefore, $(\bar{A} \kron \bar{A})^j$ must also have strictly positive entries, which means that $\bar{A} \kron \bar{A}$ must be primitive.

\section{Proof of Lemma \ref{lemma:Ppandp}}
\label{app:Ppandp}
We first show that $P_p = P_p^\T$, or equivalently,
\be
\label{eqn:pcond1}
p = \vecm(P_p^\T)
\ee

\begin{lemma}[{Vec-permutation matrix}]
\label{lemma:vecpermutation}
The $N^2\times N^2$ vec-permutation matrix $\Pi$ is a matrix whose columns are formed from the basis vectors in $\mbbR^{N^2}$ and it satisfies:
\be
\vecm(A) = \Pi \cdot \vecm(A^\T)
\ee
for any $N  \times  N$ matrix $A$. Then, for any $N  \times  N$ matrices $\{A,B\}$, 
\be
\label{eqn:AkronBPiBkronAPi}
A \kron B = \Pi (B \kron A) \Pi
\ee
In addition, $\Pi = \Pi^\T = \Pi^* = \Pi^{-1}$.
\end{lemma}
\begin{IEEEproof}
See \cite[Tabs. I and II]{Brewer78TCAS} \cite[Eqs. (5) and (6)]{Henderson81LMA}.
\end{IEEEproof}

Let $\Pi$ be the permutation matrix that satisfies
\be
\label{eqn:Pidef}
\vecm(P_p^\T) = \Pi \cdot \vecm(P_p)
\ee
From \eqref{eqn:Ppdef} and \eqref{eqn:Pidef}, proving \eqref{eqn:pcond1} is equivalent to proving
\be
\label{eqn:pcond2}
p =  \Pi \cdot p
\ee
To establish \eqref{eqn:pcond2}, we only need to show that $\Pi \cdot p$ is the Perron eigenvector of $\bar{A} \kron \bar{A} + C_A$. In that case, we can obtain \eqref{eqn:pcond2} directly from the uniqueness of the Perron eigenvector, which is $p$. Thus, note that
\begin{align}
\label{eqn:symmetric_AkA}
\Pi (\bar{A} \kron \bar{A} + C_A) \Pi & = \Pi [\E(\bm{A}_i \kron \bm{A}_i] \Pi \nn \\
& \stackrel{(a)}{=} \E\,(\bm{A}_i \kron \bm{A}_i) \nn \\
& = \bar{A} \kron \bar{A} + C_A
\end{align}
where step (a) is by \eqref{eqn:AkronBPiBkronAPi}. Then, we deduce from \eqref{eqn:pdef} that
\begin{align}
\label{eqn:Pip1}
& \Pi \cdot p = \Pi (\bar{A} \kron \bar{A} + C_A) p = (\bar{A} \kron \bar{A} + C_A) ( \Pi \cdot p) \\
\label{eqn:Pip2}
& \one_{N^2}^\T \cdot \Pi \cdot p = \one_{N^2}^\T \cdot p = 1
\end{align}
where we used the fact that $\Pi^2 = I_{N^2}$ by Lemma \ref{lemma:vecpermutation} and $\Pi \cdot \one_{N^2} = \one_{N^2}$. Results \eqref{eqn:Pip1} and \eqref{eqn:Pip2} establish that $\Pi \cdot p$ is the Perron eigenvector of $\bar{A} \kron \bar{A} + C_A$ and proves \eqref{eqn:pcond2}.

We next establish that $P_p$ is positive semi-definite. Note that for any vector $x\in\mbbR^{N}$:
\be
\label{eqn:xPpx2}
x^\T P_p x = \vecm(x^\T P_p x) = \frac{1}{N^2} (x^\T \kron x^\T) p \cdot \one_{N^2}^\T \one_{N^2}
\ee
by using \eqref{eqn:Ppdef} and the fact that $\one_{N^2}^\T \one_{N^2} = N^2$. Since $\bar{A} \kron \bar{A} + C_A = \E(\bm{A}_j \kron \bm{A}_j)$, we can introduce a series of fictitious random combination matrices $\{\bm{A}_j'; j\ge1\}$ such that they are mutually-independent and satisfy $\E(\bm{A}_j' \kron \bm{A}_j') = \bar{A} \kron \bar{A} + C_A$ for any $j\ge1$. Let $\bm{\Phi}_i \defeq \prod_{j=1}^i \bm{A}_j'$ for any $i\ge1$. Then, 
\begin{align}
\label{eqn:PhikronPhi}
\lim_{i\rightarrow\infty} \E(\bm{\Phi}_i \kron \bm{\Phi}_i) & \stackrel{(a)}{=} \lim_{i\rightarrow\infty} \prod_{j=1}^i \E(\bm{A}_j' \kron \bm{A}_j' ) \nn \\
& = \lim_{i\rightarrow\infty} (\bar{A} \kron \bar{A} + C_A)^i \nn \\
& \stackrel{(b)}{=} p \cdot \one_{N^2}^\T
\end{align}
where step (a) is by using the fact that the $\{\bm{A}_j'\}$ are mutually-independent, and step (b) is by using the Perron-Frobenius Theorem \cite{BermanPF}. Substituting \eqref{eqn:PhikronPhi} into \eqref{eqn:xPpx2} and using the fact that $\one_{N^2} = \one_N \kron \one_N$, we get
\be
\label{eqn:xPpx3}
x^\T P_p x = \frac{1}{N^2} \lim_{i\rightarrow\infty} \E\,[(x^\T \bm{\Phi}_i \one_N )^2] \ge 0
\ee
which shows that $P_p$ is positive semi-definite.

Now we show that $P_p \one_N  =  \bar{p}$. Note from \eqref{eqn:Ppdef} and \eqref{eqn:pdef} that
\be
\label{eqn:Ppandpbarexpand}
P_p = \E\,[\bm{A}_i \cdot P_p \cdot \bm{A}_i^\T]
\ee
by switching the order of $\unvecm(\cdot)$ and $\E(\cdot)$ and applying $\unvecm(\cdot)$ to the identity $\vecm(ABC) = (C^\T \kron A) \cdot \vecm(B)$. Furthermore, we get from \eqref{eqn:Ppandpbarexpand} that
\be
P_p \cdot \one_N = \E\,(\bm{A}_i P_p \bm{A}_i^\T) \cdot \one_N = \bar{A} (P_p \cdot \one_N)
\ee
which implies that the vector $P_p \cdot \one_N$ is the Perron eigenvector of $\bar{A}$, which is $\bar{p}$. Because the Perron eigenvector is unique, by \eqref{eqn:barpdef}, equation $P_p \cdot \one_N = \bar{p}$ must hold.

\section{Proof of Lemma \ref{lemma:spectralF}}
\label{app:spectralF}
We first establish that $F$ is stable if condition \eqref{eqn:meansquarestabilitycond3} is satisfied. From \eqref{eqn:Glkdef} and \eqref{eqn:Fdef}, we get
\be
\label{eqn:FequalsumGlk}
F = \sum_{\ell=1}^{N}\sum_{k=1}^{N} p_{\ell,k} G_{\ell,k}
\ee
By \eqref{eqn:pdef} and \eqref{eqn:Ppdef}, the elements $\{p_{\ell,k}\}$ of $P_p$ satisfy:
\be
\label{eqn:plkconvex}
\sum_{\ell=1}^{N}\sum_{k=1}^{N} p_{\ell,k} = 1, \quad \mbox{and} \quad p_{\ell,k} > 0
\ee
Then, in terms of the $2$-induced norm, we have
\begin{align}
\label{eqn:Fstable}
\| F \| \stackrel{(a)}{\le} \sum_{k=1}^{N} \sum_{\ell=1}^{N} p_{\ell,k} \| G_{\ell,k} \| \stackrel{(b)}{\le} \max_{k,\ell} \| G_{\ell,k} \| \stackrel{(c)}{=} \rho(\Gcal)
\end{align}
where step (a) is from the triangle inequality of norms; step (b) is by using \eqref{eqn:plkconvex}; and step (c) is by \eqref{eqn:bigGpropertyNorms}. Using \eqref{eqn:bigGpropertystable}, \eqref{eqn:Fstable}, and the fact that $\rho(F) = \| F \|$ for the Hermitian matrix $F$, we conclude that matrix $F$ is stable if condition \eqref{eqn:meansquarestabilitycond3} holds. 

We now establish expression \eqref{eqn:rhoFapprox}. Introduce the Hermitian matrix
\be
\label{eqn:Fpdef}
F' \defeq \sum_{k = 1}^{N}\sum_{\ell = 1}^{N} \bar{p}_{\ell} \bar{p}_{k} ( \bar{D}_\ell^\T\kron \bar{D}_k )
\ee
From \eqref{eqn:Dkmeandef}, we can rewrite $F'$ as $F' = (I_{2M} - H)^\T \kron (I_{2M} - H)$, where we used \eqref{eqn:barpdef} and \eqref{eqn:Hdef}. Therefore, the eigenvalues of $F'$ are equal to the products of any two of the eigenvalues of $I_{2M} - H$, which are given by $1 - \lambda(H)$. Since $\{\bar{p}_k\}$ and $\{\bar{\mu}_k\}$ are all positive and $\{ H_k \}$ are all positive definite, it is easy to verify that $H$ in \eqref{eqn:Hdef} is also positive definite. Then, from \eqref{eqn:Hdef}, \eqref{eqn:barpdef}, and \eqref{eqn:Hkdef}, and using \eqref{I-eqn:boundseigHessian} from Part I \cite{Zhao13TSPasync1} as well as Jensen's inequality \cite{Boyd04}, we get
\be
\label{eqn:Heigmax}
0 < \lambda(H) \le \| H \| \le \max_{k}\{ \bar{\mu}_k \lambda_{k,\max} \}
\ee
for all eigenvalues of $H$. When condition \eqref{eqn:meansquarestabilitycond3} holds, we get from \eqref{I-eqn:boundmumoments0} of Part I \cite{Zhao13TSPasync1} that
\be
\bar{\mu}_k \le \frac{\bar{\mu}_k^{(2)}}{\bar{\mu}_k} <  \frac{\lambda_{k,\min}}{\lambda_{k,\max}^2 + \alpha} < \frac{1}{\lambda_{k,\max}}
\ee
for any $k$. This implies that $\max_{k} \{ \bar{\mu}_k\lambda_{k,\max} \} < 1$ and therefore, $0 < \lambda(H) < 1$ for all eigenvalues of $H$. From \eqref{eqn:Hdef} and \eqref{eqn:bigMbound}, we obtain
\be
\label{eqn:boundlamdbaH1}
0 < \lambda(H) = O(\nu) < 1
\ee
for any eigenvalue of $H$. Therefore, we get 
\be
\label{eqn:rhoIminusH}
\lambda(I_{2M} - H) = 1 - O(\nu), \;\; 
\rho(I_{2M} - H) = 1 - \lambda_{\min}(H)
\ee
where $\lambda_{\min}(\cdot)$ denotes the smallest eigenvalue of its Hermitian matrix argument. We further get from \eqref{eqn:rhoIminusH} that
\be
\label{eqn:lambdaFpbound}
\lambda(F') = 1 - O(\nu), \quad 
\rho(F') = [1 - \lambda_{\min}(H)]^2
\ee
Using Lemma \ref{lemma:Ppandp}, \eqref{eqn:pdef}, \eqref{eqn:barpdef}, and \eqref{eqn:Dkmeandef}, the difference between $F$ in \eqref{eqn:Fdef} and $F'$ in \eqref{eqn:Fpdef} is given by
\be
F - F' = \sum_{k=1}^{N} \sum_{\ell=1}^{N} \{ [ (p_{\ell,k} - \bar{p}_{\ell} \bar{p}_k) \bar{\mu}_\ell \bar{\mu}_k + p_{\ell,k} c_{\mu,\ell,k}] (H_\ell^\T \kron H_k) \} 
\ee
which is also Hermitian. From \eqref{I-eqn:boundmumoments2} in Part I \cite{Zhao13TSPasync1}, we get
\begin{align}
\label{eqn:bigMbound}
\bar{\mu}_k & \equiv \bar{\mu}_k^{(1)} \le \nu \\
\label{eqn:bigCMbound1}
c_{\mu,k,k} & \le \bar{\mu}_k^{(2)} \le \nu^2 \\
\label{eqn:bigCMbound2}
|c_{\mu,\ell,k}| & \le \sqrt{c_{\mu,\ell,\ell} \cdot c_{\mu,k,k}} \le \nu^2
\end{align}
where \eqref{eqn:bigCMbound2} is by using the Cauchy-Schwartz inequality. By \eqref{eqn:bigMbound}--\eqref{eqn:bigCMbound2}, we get $\| F - F' \| = O(\nu^2)$. Using a corollary of the Wielandt-Hoffman theorem \cite[Corollary~8.1.6, p.~396]{Golub96}, we then conclude that
\be
\label{eqn:pairedeigenvalues}
| \lambda_m(F) - \lambda_m(F') | \le \| F - F' \| = O(\nu^2)
\ee
where $\lambda_m(\cdot)$ denotes the $m$-th eigenvalue of its Hermitian matrix argument; the eigenvalues are assumed to be ordered from largest to smallest in each case. Result \eqref{eqn:pairedeigenvalues} implies that for every eigenvalue of $F'$ there is an eigenvalue of $F$ that is $O(\nu^2)$ close to it. From \eqref{eqn:pairedeigenvalues} and \eqref{eqn:lambdaFpbound} we immediately deduce that
\be
\label{eqn:lambdaFclosetolambdaFp}
\lambda_m(F) = 1 - O(\nu), \quad 
\rho(F) = \rho(F') + O(\nu^2)
\ee
where $\rho(F')$ from \eqref{eqn:lambdaFpbound} dominates the $O(\nu^2)$ term.

\section{Proof of Lemma \ref{lemma:lowrank}}
\label{app:lowrank}
We first establish \eqref{eqn:lowrankapproxIminusF}. Introduce the Jordan decomposition:
\be
\label{eqn:AACAeig}
\bar{A} \kron \bar{A} + C_A \defeq P J Q^\T = \begin{bmatrix}
p & P'
\end{bmatrix} \begin{bmatrix}
1 & 0 \\
0 & J' \\
\end{bmatrix} \begin{bmatrix}
\one_{N^2} & Q'
\end{bmatrix}^\T
\ee
where $J$ is the Jordan canonical form of $\bar{A} \kron \bar{A} + C_A$ and $J'$ is a sub-matrix of $J$ containing its stable eigenvalues, $P'$ and $Q'$ are sub-matrices of $P$ and $Q$, and $P^{-1} = Q^\T$. Then, the Jordan decomposition of $\bar{\Acal} \bkron \bar{\Acal} + \Ccal_A$ is given by
\be
\label{eqn:bigAAEVD}
\bar{\Acal} \bkron \bar{\Acal} + \Ccal_A = \Pcal \Jcal \Qcal^\T = \begin{bmatrix}
\Pcal_1 & \Pcal'
\end{bmatrix} \begin{bmatrix}
I_{4M^2} & 0 \\
0 & \Jcal' \\
\end{bmatrix} \begin{bmatrix}
\Qcal_1 & \Qcal'
\end{bmatrix}^\T
\ee
where
\begin{alignat}{3}
\label{eqn:bigPdef}
\Pcal & \defeq P \kron I_{4M^2}, & \qquad \Pcal' & \defeq P' \kron I_{4M^2} \\
\label{eqn:bigJdef}
\Jcal & \defeq J \kron I_{4M^2}, & \qquad \Jcal' & \defeq J' \kron I_{4M^2} \\
\label{eqn:bigQdef}
\Qcal & \defeq Q \kron I_{4M^2}, & \qquad \Qcal' & \defeq Q' \kron I_{4M^2} \\
\label{eqn:pbarandqbardef}
\Pcal_1 & \defeq p \kron I_{4M^2}, & \qquad \Qcal_1 & \defeq \one_{N^2} \kron I_{4M^2}
\end{alignat}
Let
\be
\label{eqn:bigXdef}
\Xcal \defeq I_{4M^2N^2} - \Gcal
\ee
where $\Gcal$ is given by \eqref{eqn:bigGdef}. Then, by \eqref{eqn:bigFdef},
\begin{align}
\label{eqn:QFP}
\Qcal^\T\Fcal\Pcal & = \Qcal^\T [ \Gcal \cdot (\bar{\Acal} \bkron \bar{\Acal} + \Ccal_A) ] \Pcal \nn \\
& = \Qcal^\T (I_{4M^2N^2} - \Xcal) (\Pcal \Jcal \Qcal^\T) \Pcal \nn \\
& = \Jcal - \Qcal^\T \Xcal \Pcal \Jcal \nn \\
& = \begin{bmatrix}
I_{4M^2} - \Qcal_1^\T \Xcal \Pcal_1 & -\Qcal_1^\T \Xcal \Pcal' \Jcal' \\
-\Qcal'^\T \Xcal \Pcal_1 & \Jcal' - \Qcal'^\T \Xcal \Pcal' \Jcal' \\
\end{bmatrix}
\end{align}
From \eqref{eqn:QFP}, we further get
\be
\label{eqn:QIFP1}
( I_{4M^2N^2}  -  \Qcal^\T\Fcal\Pcal )^{-1} = \begin{bmatrix}
\Qcal_1^\T \Xcal \Pcal_1  & \Qcal_1^\T \Xcal \Pcal' \Jcal' \\
\Qcal'^\T \Xcal \Pcal_1  & I   -   \Jcal'   +   \Qcal'^\T \Xcal \Pcal' \Jcal' \\
\end{bmatrix}^{-1}
\ee
where the $I$ denotes the $4M^2(N^2-1)\times 4M^2(N^2-1)$ identity matrix. The quantity $\Qcal_1^\T \Xcal \Pcal_1$ in \eqref{eqn:QFP} can be expressed as
\begin{align}
\label{eqn:qXpexpression}
\Qcal_1^\T \Xcal \Pcal_1 & \stackrel{(a)}{=} \Qcal_1^\T \cdot \Pcal_1 - \Qcal_1^\T \Gcal \Pcal_1 \nn \\
& \stackrel{(b)}{=} (\one_{N^2}^\T p) \kron I_{4M^2} - (\one_{N^2}^\T \kron I_{4M^2}) ( \diag\{ G_{\ell,k} \} ) (p \kron I_{4M^2}) \nn \\
& \stackrel{(c)}{=} I_{4M^2} - F
\end{align}
where step (a) is by \eqref{eqn:bigXdef}; step (b) is by \eqref{eqn:bigGpropertyBD1}--\eqref{eqn:bigGpropertyBD2}; and step (c) is by \eqref{eqn:FequalsumGlk}. We already know that the matrices $F$ and $\Fcal$ are stable for sufficiently small step-sizes. Thus, the matrices $I_{4M^2N^2}  -  \Fcal$ and $I_{4M^2}  -  F$ are invertible. It follows that the quantity  $\Qcal_1^\T \Xcal \Pcal_1$ is invertible. Moreover, the Schur complement with respect to $\Qcal_1^\T \Xcal \Pcal_1$ in \eqref{eqn:QIFP1} is also invertible. Let us denote the inverse of this Schur complement by
\begin{align}
\label{eqn:Deltadef}
\Delta & \defeq [I-\Jcal' + \Qcal'^\T \Xcal \Pcal' \Jcal'
 - \Qcal'^\T \Xcal \Pcal_1 (\Qcal_1^\T \Xcal \Pcal_1)^{-1} \Qcal_1^\T \Xcal \Pcal' \Jcal' ]^{-1}
\end{align}
Then, by using a formula for the inversion of block matrices \cite[Eq.~(7), p.~48]{Laub05}, equality \eqref{eqn:QIFP1} can be expressed as
\begin{align}
\label{eqn:QIFP2}
(I_{4M^2N^2}  -  \Qcal^\T \Fcal \Pcal)^{-1} & = \begin{bmatrix}
(\Qcal_1^\T \Xcal \Pcal_1)^{-1}  +  \Delta'    &    -(\Qcal_1^\T \Xcal \Pcal_1)^{-1} \Qcal_1^\T \Xcal \Pcal' \Jcal' \Delta \\
-\Delta \Qcal'^\T \Xcal \Pcal_1 ( \Qcal_1^\T \Xcal \Pcal_1)^{-1}     &    \Delta \\
\end{bmatrix}     
\end{align}
where
\be
\label{eqn:Deltaprimedef}
\Delta' \defeq (\Qcal_1^\T \Xcal \Pcal_1)^{-1} \Qcal_1^\T \Xcal \Pcal' \Jcal'  \Delta \Qcal'^\T \Xcal \Pcal_1 (\Qcal_1^\T \Xcal \Pcal_1)^{-1}
\ee
Now, from \eqref{eqn:qXpexpression}, \eqref{eqn:Fdef}, \eqref{eqn:Dkmeandef}, \eqref{eqn:plkandpklandbarpk}, and \eqref{eqn:Hdef}, we can also write
\begin{align}
\label{eqn:qXpexpression2}
\Qcal_1^\T \Xcal \Pcal_1 & = H \kron I_{2M} + I_{2M} \kron H  - \sum_{\ell,k=1}^{N} p_{\ell,k} (\bar{\mu}_\ell \bar{\mu}_k + c_{\mu,\ell,k})(H_\ell^\T \kron H_k)
\end{align}
It follows from \eqref{eqn:qXpexpression2} and \eqref{eqn:bigMbound}--\eqref{eqn:bigCMbound2} that $\Qcal_1^\T \Xcal \Pcal_1$ is Hermitian and
\be
\label{eqn:invqXporder}
\| \Qcal_1^\T \Xcal \Pcal_1 \| = O(\nu), \quad \| (\Qcal_1^\T \Xcal \Pcal_1)^{-1} \| = O(\nu^{-1})
\ee
Likewise, from \eqref{eqn:bigMmeandef} and \eqref{eqn:bigMbound}--\eqref{eqn:bigCMbound2}, we get that
\be
\label{eqn:bigMandbigCMorder}
\| \bar{\Mcal} \| = O(\nu), \qquad \| \Ccal_M \| = O(\nu^2)
\ee
and from \eqref{eqn:bigXdef}, \eqref{eqn:bigGdef}, \eqref{eqn:bigDmeandef}, and \eqref{eqn:bigCDdef}, we further get 
\begin{align}
\label{eqn:bigXorder}
\| \Xcal \| & = \| (\bar{\Mcal}\Hcal)^\T \bkron I_{2MN} + I_{2MN} \bkron (\bar{\Mcal}\Hcal) + O(\nu^2) \|  = O(\nu)
\end{align}
since matrix $\Hcal$ is constant and independent of $\nu$. Furthermore, it follows from \eqref{eqn:bigXorder} that
\begin{align}
\label{eqn:Deltaparts}
\| \Qcal'^\T \Xcal \Pcal' \Jcal' \| = O(\nu), \quad \| \Qcal'^\T \Xcal \Pcal_1 \| = O(\nu), \quad
\| \Qcal_1^\T \Xcal \Pcal' \Jcal'\| = O(\nu)
\end{align}
From \eqref{eqn:bigAAEVD}, matrix $I - \Jcal'$ is invertiable and is independent of $\nu$. Therefore,
\be
\label{eqn:IminusbigJpandinverse}
\| \Jcal' \| = O(1), \quad \| I - \Jcal'\| = O(1), \quad \| (I - \Jcal')^{-1} \| = O(1)
\ee
Then, by \eqref{eqn:Deltadef}, \eqref{eqn:Deltaparts}, and \eqref{eqn:IminusbigJpandinverse}, we get
\be
\label{eqn:Deltaorder}
\| \Delta \| = \| (I - \Jcal' + O(\nu))^{-1} \| = O(1)
\ee
By \eqref{eqn:Deltaprimedef}, \eqref{eqn:invqXporder}, \eqref{eqn:Deltaparts}, and \eqref{eqn:Deltaorder}, we further get
\begin{align}
\label{eqn:Deltaporder}
\| \Delta' \| = O(1), \quad \| (\Qcal_1^\T \Xcal \Pcal_1)^{-1} \Qcal_1^\T \Xcal \Pcal' \Jcal' \Delta \| = O(1), \quad \| \Delta \Qcal'^\T \Xcal \Pcal_1 ( \Qcal_1^\T \Xcal \Pcal_1)^{-1} \| = O(1) 
\end{align}
Using \eqref{eqn:Deltaporder} and Assumption \ref{asm:smallstepsizes}, we get from \eqref{eqn:QIFP2} that
\be
\label{eqn:IminusFcalinv}
(I_{4M^2N^2} - \Qcal^\T\Fcal\Pcal)^{-1} = \begin{bmatrix}
(\Qcal_1^\T \Xcal \Pcal_1)^{-1}  & 0 \\
0  & 0 \\
\end{bmatrix} + O(1)
\ee
Then, 
\begin{align}
\label{eqn:IminusFcalinv1}
(I_{4M^2N^2}   -   \Fcal)^{-1} & \stackrel{(a)}{=} \Pcal\begin{bmatrix}
(\Qcal_1^\T \Xcal \Pcal_1)^{-1}  & 0 \\
0  & 0 \\
\end{bmatrix}\Qcal^\T + O(1) \nn \\
& \stackrel{(b)}{=} \Pcal_1(\Qcal_1^\T \Xcal \Pcal_1)^{-1}\Qcal_1^\T + O(1) \nn \\
& \stackrel{(c)}{=} (p \one_{N^2}^\T) \kron (I_{4M^2}   -   F)^{-1}   +   O(1)  
\end{align}
where step (a) is by using the fact that $\Pcal^{-1} = \Qcal^\T$; step (b) is by using the block division in \eqref{eqn:bigAAEVD}; step (c) is by using \eqref{eqn:pbarandqbardef}; and by \eqref{eqn:qXpexpression} and  \eqref{eqn:invqXporder},
\begin{align}
\label{eqn:lowrankapproxIminusF1}
\| (I_{4M^2} - F)^{-1} \| = O(\nu^{-1})
\end{align}
Under Assumption \ref{asm:smallstepsizes}, the parameter $\nu  \ll  1$. Therefore, $\nu^{-1}  \gg  1$ and $(p \one_{N^2}^\T) \kron (I_{4M^2}  -  F)^{-1}$ dominates the $O(1)$ term in \eqref{eqn:IminusFcalinv1}.

Finally, we establish result \eqref{eqn:rhobigFapprox}. Let
\be
\label{eqn:bigFnew}
\Fcal_s \defeq \Qcal^\T \Fcal \Pcal = \begin{bmatrix}
F & O(\nu) \\
O(\nu) & \Jcal' + O(\nu) \\
\end{bmatrix}
\ee
by using \eqref{eqn:QFP}, \eqref{eqn:qXpexpression}, \eqref{eqn:Deltaparts}, and \eqref{eqn:IminusbigJpandinverse}. Since $\Fcal_s$ is similar to $\Fcal$, they have the same eigenvalues \cite{Laub05}. Since $F$ is Hermitian, let us introduce its eigenvalue decomposition as
\be
\label{eqn:Feigdef}
F = U \Lambda U^*
\ee
where $U$ is a $4M^2\times 4M^2$ unitary matrix and $\Lambda$ is a $4M^2\times 4M^2$ diagonal matrix. The $(N^2-1)\times(N^2-1)$ matrix $J'$, which contains the stable eigenvalues of $\bar{A} \kron \bar{A} + C_A$ in \eqref{eqn:AACAeig}, can be generally expressed as
\be
\label{eqn:Jpdef}
J' = \begin{bmatrix}
\lambda_{a,2} &        & T'             \\
              & \ddots &               \\
0             &        & \lambda_{a,N^2} \\
\end{bmatrix}
\ee
where $\{\lambda_{a,n}\}$ are the eigenvalues of $\bar{A} \kron \bar{A} + C_A$ with $\lambda_{a,1} = 1$ and $|\lambda_{a,n}| < 1$ for all $n = 2,3,\dots, N^2$. In \eqref{eqn:Jpdef}, the elements in the strictly upper triangular region $T'$ are either 1 or 0, which depend on the Jordan blocks in $J'$. Using \eqref{eqn:Jpdef} and \eqref{eqn:bigJdef}, we can express the $(2,2)$ block in \eqref{eqn:bigFnew} as
\be
\label{eqn:bigJpdef}
\Jcal'  +  O(\nu)  =  \begin{bmatrix}
\lambda_{a,2}I_{4M^2}  +  O(\nu)   &            &   \Tcal'  +  O(\nu) \\
                &   \ddots   &                 \\
O(\nu)               &            &   \lambda_{a,N^2}I_{4M^2}  +  O(\nu)  \\
\end{bmatrix}  
\ee
where the elements in the strictly upper triangular region $\Tcal'$ are either 1 or 0, which depend on the elements of $T'$ in \eqref{eqn:Jpdef}. We now apply a similarity transformation to $\Fcal$ by multiplying
\be
\Dcal \defeq \diag\{ \nu^{\epsilon} U, \nu^{2\epsilon} I_{4M^2}, \nu^{3\epsilon}I_{4M^2}, \dots, \nu^{N^2\epsilon}I_{4M^2} \}
\ee
and its inverse $\Dcal^{-1}$ on either side of \eqref{eqn:bigFnew}, where $\epsilon = 1/N^2$. Using \eqref{eqn:bigFnew} and \eqref{eqn:bigJpdef}, we end up with
\begin{align}
\label{eqn:bigFnew1}
\Dcal^{-1} \Fcal_s \Dcal  d& = \left[\begin{array}{c|c}
 \Lambda               &   O(\nu^{1+\epsilon}) \\
\hline
 O(\nu^{\epsilon})   &    \begin{array}{ccc}
\lambda_{a,2}I_{4M^2}   +   O(\nu)    &         &    O(\nu^{\epsilon}) \\
   &    \ddots        &       \\
O(\nu^{\epsilon})    &        &    \lambda_{a,N^2}I_{4M^2}   +  O(\nu)    \\
\end{array} \\
\end{array}  
\right]
\end{align}
From \eqref{eqn:bigFnew1}, we know that all off-diagonal entries of $\Dcal^{-1} \Fcal_s \Dcal$ are \emph{at least} of the order of $\nu^\epsilon$. Therefore, using Gershgorin Theorem \cite[p.~320]{Golub96} under Assumption \ref{asm:smallstepsizes}, and since $\Fcal$ and $\Fcal_s$ have the same eigenvalues due to similarity, we get
\be
\label{eqn:rhobigFandrhobigYclose}
| \lambda(\Fcal)  -  \lambda(F) | \le O(\nu^{1+\epsilon}) \;\; \mbox{or} \;\;
| \lambda(\Fcal)  -  \lambda_{a,k} | \le O(\nu^{\epsilon})
\ee
where $\lambda(\Fcal)$ denotes the eigenvalue of $\Fcal$ and $k = 2, 3, \dots, N^2$. Result \eqref{eqn:rhobigFandrhobigYclose} implies that the eigenvalues of $\Fcal$ are either located in the Gershgorin circles that are centered at the eigenvalues of $F$ with radii $O(\nu^{1+\epsilon})$ or in the Gershgorin circles that are centered at $\{\lambda_{a,k}; k=2,3,\dots,N^2\}$ with radii $O(\nu^{\epsilon})$. From \eqref{eqn:lambdaFclosetolambdaFp}, we have
\be
\label{eqn:rhoFandlambdamaxF}
\rho(F) = 1 - O(\nu) < 1
\ee
By Assumption \ref{asm:connected} and Perron-Frobenius Theorem \cite{BermanPF}, we have
\be
\label{eqn:rhorelation2}
\rho(J') \defeq \max_{k=2,3,\dots,N^2} | \lambda_{a,k} | < 1
\ee
By Assumption \ref{asm:smallstepsizes}, if the parameter $\nu$ is small enough to satisfy
\be
O(\nu^\epsilon) + O(\nu) < 1 - \rho(J')
\ee
such that the inequality
\be
\label{eqn:lambdaFgerhoJp}
\rho(J') + O(\nu^\epsilon) < 1 - O(\nu) = \rho(F)
\ee
holds, then the Gershgorin circles centered at the eigenvalues of $F$ are isolated from those centered at $\{\lambda_{a,k}; k=2,3,\dots,N^2\}$. According to Gershgorin Theorem \cite[p.~181]{Stewart90}, there are precisely $4M^2$ eigenvalues of $\Fcal$ satisfying
\be
\label{eqn:rhobigFclosetolambdaF}
| \lambda(\Fcal) - \lambda(F) | \le O(\nu^{1+\epsilon})
\ee
while all the other eigenvalues satisfy
\be
\label{eqn:rhobigFclosetolambdaak}
| \lambda(\Fcal) - \lambda_{a,k} | \le O(\nu^{\epsilon}), \quad k = 2, 3, \dots, N^2
\ee
By \eqref{eqn:lambdaFgerhoJp}, the eigenvalues $\lambda(\Fcal)$ satisfying \eqref{eqn:rhobigFclosetolambdaF} are greater than those satisfying \eqref{eqn:rhobigFclosetolambdaak} in magnitude. Furthermore, when $\nu$ is sufficiently small, the Gershgorin circles centered at $\lambda_{\max}(F)$ with radius $O(\nu^{1+\epsilon})$ will become disjoint from the other circles (see Fig. \ref{fig:eigenvalues}). Then, by using \eqref{eqn:rhoFandlambdamaxF} and Gershgorin Theorem again, we conclude from \eqref{eqn:rhobigFclosetolambdaF} that
\be
\rho(\Fcal) = \rho(F) + O(\nu^{1+\epsilon})
\ee
It is worth noting that from \eqref{eqn:rhoFapprox} we get
\be
\rho(\Fcal) = 1 - O(\nu) + O(\nu^{1+\epsilon}) < 1
\ee
for $\nu \ll 1$ because $\epsilon = 1/N^2 > 1$.

\begin{figure}
\includegraphics[scale=0.7]{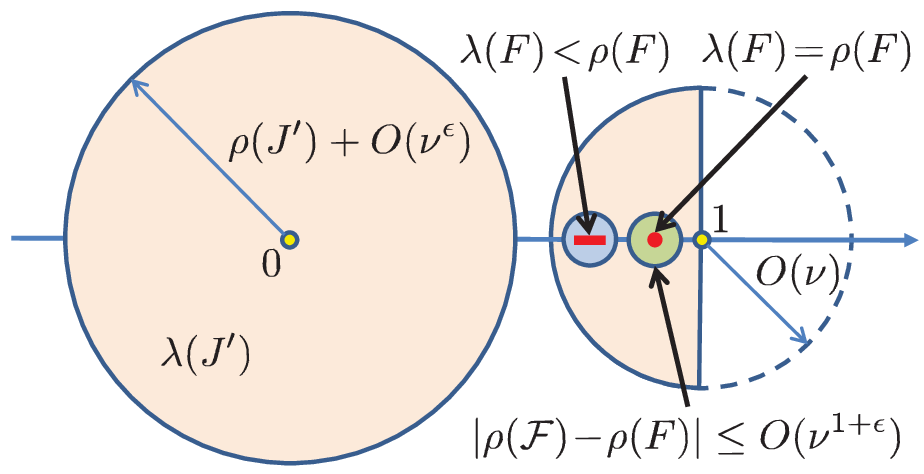}
\centering
\caption{An illustration of the locations of the eigenvalues of $\Fcal$. The eigenvalues of $J'$ are all in the left big circle, so the eigenvalues of $\Fcal$ satisfying \eqref{eqn:rhobigFclosetolambdaak} are also in the left big circle. The eigenvalues of $F$ are all in the right big circle, so the eigenvalues of $\Fcal$ satisfying \eqref{eqn:rhobigFclosetolambdaF} are also in the right big circle. Specifically, the eigenvalues of $F$ with $\lambda(F) < \rho(F)$ are all on the red segment on the horizontal line, so the eigenvalues of $\Fcal$ that satisfy \eqref{eqn:rhobigFclosetolambdaF} are all in the small blue circle on the left; the eigenvalues of $F$ with $\lambda(F) = \rho(F)$ are on the red dot on the horizontal line, so the eigenvalues of $\Fcal$ that satisfy \eqref{eqn:rhobigFclosetolambdaF} are all in the small green circle on the right.}
\label{fig:eigenvalues}
\vspace{-1\baselineskip}
\end{figure}

\section{Proof of Theorem \ref{theorem:networkcovariance}}
\label{app:networkcovariance}
From \eqref{eqn:bigAAEVD} and \eqref{eqn:IminusFcalinv1}, we first have
\begin{align}
\label{eqn:part1approx}
(\bar{\Acal} \bkron \bar{\Acal} + \Ccal_A) (I_{4M^2N^2} - \Fcal)^{-1} & = \begin{bmatrix}
\Pcal_1  &  \Pcal'
\end{bmatrix} \begin{bmatrix}
I_{4M^2}  &  0 \\
0  &  \Jcal' \\
\end{bmatrix} \begin{bmatrix}
(\Qcal_1^\T \Xcal \Pcal_1)^{-1}   &  0 \\
0  &  0 \\
\end{bmatrix} \begin{bmatrix}
\Qcal_1  &  \Qcal'
\end{bmatrix}^\T  +  O(1)\nn \\
& = \Pcal_1 (\Qcal_1^\T \Xcal \Pcal_1)^{-1} \Qcal_1^\T + O(1)
\end{align}
Since $\Pcal_1 (\Qcal_1^\T \Xcal \Pcal_1)^{-1} \Qcal_1^\T  =  O(\nu^{-1})$ by \eqref{eqn:invqXporder}, the first term on the RHS of \eqref{eqn:part1approx} dominates the second term under Assumption \ref{asm:smallstepsizes}. By \eqref{eqn:bigMmeandef}, \eqref{eqn:bigRdef}, \eqref{eqn:bigCMdef}, and \eqref{eqn:bvecdef}, we get
\begin{align}
\label{eqn:vecRequ}
\Pcal_1^\T (\bar{\Mcal} \bkron \bar{\Mcal} + \Ccal_M) \bvec(\Rcal) & = (p^\T \kron I_{4M^2}) [(\bar{M} \kron \bar{M} + C_M) \kron I_{4M^2}] \bvec(\Rcal) \nn \\
& = \vecm(R)
\end{align}
where $R$ is defined by \eqref{eqn:Rdef} and is of the order of $\nu^2$. We then get a low-rank expression for $z$ in \eqref{eqn:bigzdef}:
\begin{align}
\label{eqn:bigzapprox}
z & \stackrel{(a)}{=} [\Pcal_1 (\Qcal_1^\T \Xcal \Pcal_1)^{-1} \Qcal_1^\T + O(1)]^* (\bar{\Mcal} \bkron \bar{\Mcal} + \Ccal_M) \bvec(\Rcal) \nn \\
& \stackrel{(b)}{=} \Qcal_1 [ (\Qcal_1^\T \Xcal \Pcal_1)^{-1} ]^* \Pcal_1^\T (\bar{\Mcal} \bkron \bar{\Mcal}  +  \Ccal_M) \bvec(\Rcal)  +  O(\nu^2) \nn \\
& \stackrel{(c)}{=} \Qcal_1 [(I_{4M^2} - F)^{-1}]^* \vecm(R) + O(\nu^2) \nn \\
& \stackrel{(d)}{=} \one_{N^2} \kron [(I_{4M^2} - F)^{-1} \vecm(R) ] + O(\nu^2)
\end{align}
where step (a) is by \eqref{eqn:part1approx}; step (b) is by \eqref{eqn:boundMMCM}; step (c) is by \eqref{eqn:qXpexpression} and \eqref{eqn:vecRequ}; and step (d) is by \eqref{eqn:pbarandqbardef} and the fact that $F$ is Hermitian. The first term on the RHS of \eqref{eqn:bigzapprox} is dominant due to  \eqref{eqn:lowrankapproxIminusF1} and \eqref{eqn:orderofRandZ}. Applying $\unbvec(\cdot)$ to both sides of \eqref{eqn:bigzapprox} and using \eqref{eqn:Zdef} yields
\be
\label{eqn:bigZapprox}
\unbvec( z )  = (\one_N \one_N^\T) \kron Z + O(\nu^2)
\ee
where $Z$ is given by \eqref{eqn:Zdef} and is of the order of $\nu$ by \eqref{eqn:orderofRandZ}. From \eqref{eqn:bigzidef} and \eqref{eqn:sinftydef}, we know that $\unbvec(z_\infty)$ is the steady-state covariance matrix of $\ubar{\wt{\bm{w}}}_i'$. Using \eqref{eqn:zinftyandz} and \eqref{eqn:bigZapprox}, the steady-state covariance matrix of $\ubar{\wt{\bm{w}}}_i'$ can be approximated by
\begin{align}
\label{eqn:approximatesteadystateCov0}
\lim_{i\rightarrow\infty} \E \ubar{\wt{\bm{w}}}_i' \ubar{\wt{\bm{w}}}_i'^* & = \unbvec(z_\infty) \nn \\
& = (\one_N \one_N^\T) \kron Z + O(\nu^{1 + \min\{ 2, \gamma_v\} /2 }) \nn \\
& = (\one_N \one_N^\T) \kron Z + O(\nu^{1 + \gamma_o' })
\end{align}
where $\gamma_o'$ is given by \eqref{eqn:gammaopdef}, and the first term on the RHS is dominant due to \eqref{eqn:orderofRandZ}.

\section{Proof of Lemma \ref{lemma:propertiesZ}}
\label{app:propertiesZ}
From Lemma \ref{lemma:spectralF} and Assumption \ref{asm:smallstepsizes}, we know that matrix $F$ is stable. From \eqref{eqn:Zdef}, we get
\be
\label{eqn:Zseriesexpression}
Z  =  \unvecm\left( \sum_{j=0}^{\infty} F^j \vecm(R) \right)  =  \sum_{j=0}^{\infty} \unvecm\left( F^j \vecm(R) \right)  
\ee
Let
\be
\label{eqn:Rjdef}
R^{(j)} \defeq \unvecm\left( F^j \vecm(R) \right), \qquad j\ge0
\ee
where $R^{(0)} = R$. Then, since $F$ is stable, the $2M\times 2M$ matrix sequence $\{R^{(j)}; j\ge0\}$ converges to zero. Substituting \eqref{eqn:Rjdef} into \eqref{eqn:Zseriesexpression} yields
\be
\label{eqn:ZandRj}
Z = \sum_{j=0}^{\infty} R^{(j)}
\ee

\begin{lemma}[{Condition for complex-Hessian-type matrices}]
\label{lemma:complexhessian}
A sufficient and necessary condition for any $2M\times 2M$ positive semi-definite matrix $H$ to be a complex-Hessian-type matrix in Definition \ref{def:complexHessian} is to require $L H^\T L = H$, where 
\be
L \defeq \begin{bmatrix}
0 & I_M \\
I_M & 0 \\
\end{bmatrix}
\ee
satisfies $L = L^\T = L^{-1}$.
\end{lemma}
\begin{IEEEproof}
Let the $2M\times 2M$ positive semi-definite matrix be
\be
\label{eqn:Hpartitian}
H = \begin{bmatrix}
A & B \\
B^* & D \\
\end{bmatrix}
\ee
where $\{A,B,D\}$ are $M\times M$ submatrices satisfying $A = A^*$ and $D = D^*$. Then,
\be
L H^\T L \defeq \begin{bmatrix}
D^\T & B^\T \\
(B^*)^\T & A^\T \\
\end{bmatrix}
\ee
By Definition \ref{def:complexHessian}, the matrix $H$ is a complex-Hessian-type matrix if, and only if, $A = D^\T$ and $B = B^\T$. It is straightforward to verify that these conditions are equivalent to the equality $LH^\T L = H$.
\end{IEEEproof}

Using Lemma \ref{lemma:complexhessian}, it is easy to verify that if each $R^{(j)}$, $j\ge0$, in \eqref{eqn:ZandRj} is Hermitian positive semi-definite and of complex-Hessian-type, then so is $Z$. Now, from \eqref{eqn:Rjdef}, we have
\be
\vecm\left( R^{(j)} \right)  =  F^j \vecm(R)  =  F F^{j-1} \vecm(R)  =  F \vecm\left( R^{(j-1)} \right)
\ee
From \eqref{eqn:Fdef} and using the property $\vecm(ABC) = (C^\T \kron A) \vecm(B)$, we get the following recursion:
\be
\label{eqn:Rjrecursive}
R^{(j)}  =   \sum_{k = 1}^{N}\sum_{\ell = 1}^{N} p_{\ell,k} [\bar{D}_k R^{(j-1)} \bar{D}_\ell   +   c_{\mu,\ell,k} H_k R^{(j-1)} H_\ell ]   
\ee
We can now verify by mathematical induction that each $R^{(j)}$ is Hermitian positive semi-definite and of complex-Hessian-type. Obviously, from \eqref{eqn:Rdef}, we know that $R^{(0)} = R$ is Hermitian positive semi-definite and of complex-Hessian-type. Now,  assuming that $R^{(j-1)}$ is Hermitian positive semi-definite and of complex-Hessian-type, let us verify that the same applies to $R^{(j)}$.

Since $\{\bar{D}_k, H_k\}$ are all Hermitian matrices, it is easy to verify that
\begin{align}
\label{eqn:Rjpart1}
\sum_{k = 1}^{N}\sum_{\ell = 1}^{N} p_{\ell,k} \bar{D}_k R^{(j-1)} \bar{D}_\ell 
& = \begin{bmatrix}
\bar{D}_1 \\
\vdots \\
\bar{D}_N \\
\end{bmatrix}^* \begin{bmatrix}
p_{1,1} R^{j-1} & \dots & p_{1,N} R^{j-1} \\
\vdots & \ddots & \vdots \\
p_{N,1} R^{j-1} & \dots & p_{N,N} R^{j-1} \\
\end{bmatrix} \begin{bmatrix}
\bar{D}_1 \\
\vdots \\
\bar{D}_N \\
\end{bmatrix} \nn \\
& = (\one_N \kron I_{2M})^* \bar{\Dcal}^* [P_p \kron R^{(j-1)}] \bar{\Dcal} (\one_N \kron I_{2M})
\end{align}
and
\begin{align}
\label{eqn:Rjpart2}
\sum_{k = 1}^{N}\sum_{\ell = 1}^{N} p_{\ell,k} c_{\mu,\ell,k} H_k R^{(j-1)} H_\ell  
& = \begin{bmatrix}
H_1 \\
\vdots \\
H_N \\
\end{bmatrix}^* \begin{bmatrix}
p_{1,1} c_{\mu,1,1} R^{j-1} \!&\! \dots \!&\! p_{1,N} c_{\mu,1,N} R^{j-1} \\
\vdots & \ddots & \vdots \\
p_{N,1} c_{\mu,N,1} R^{j-1} \!&\! \dots \!&\! p_{N,N} c_{\mu,N,N} R^{j-1} \\
\end{bmatrix} \begin{bmatrix}
H_1 \\
\vdots \\
H_N \\
\end{bmatrix} \nn \\
& = (\one_N \kron I_{2M})^* \Hcal^*  [(P_p \odot C_{\mu} ) \kron R^{(j-1)}] \Hcal (\one_N \kron I_{2M})
\end{align}
where $\bar{\Dcal}$ is from \eqref{eqn:bigDmeandef}, $\Hcal$ is from \eqref{eqn:errorrecursiondefapprox}, $C_{\mu} \defeq [c_{\mu,\ell,k}]_{\ell,k = 1}^{N}$, and $\odot$ denotes the Hadamard product (the element-wise product) of matrices \cite{Horn91}. Since $P_p$ is Hermitian positive semi-definite by Lemma \ref{lemma:Ppandp}, and $R^{(j-1)}$ is also Hermitian  positive semi-definite by the induction hypothesis, the Kronecker product $P_p \kron R^{(j-1)}$ must be Hermitian  positive semi-definite \cite[p.~245]{Horn91}. Therefore, the term on the LHS of \eqref{eqn:Rjpart1} must be Hermitian  positive semi-definite. From \eqref{I-eqn:randcombinecoventry} of Part I \cite{Zhao13TSPasync1}, it is obvious that the $C_\mu$ is the covariance matrix of $\{\bm{\mu}_k(i)\}$ and it must be Hermitian positive semi-definite. Since $P_p$ and $C_\mu$ are both Hermitian positive semi-definite, the Hadamard product $P_p \odot C_{\mu}$ is also Hermitian positive semi-definite by the Schur product Theorem \cite[p.~309]{Horn91}. Then, the Kronecker product $(P_p \odot C_{\mu} ) \kron R^{(j-1)}$ must be Hermitian positive semi-definite \cite[p.~245]{Horn91}, and in turn, the term on the LHS of \eqref{eqn:Rjpart2} must also be Hermitian positive semi-definite. From \eqref{eqn:Rjpart1} and \eqref{eqn:Rjpart2}, we conclude that the $R^{(j)}$ in \eqref{eqn:Rjrecursive} must be Hermitian positive semi-definite.

Finally, we show that if $R^{(j-1)}$ is of complex-Hessian-type, then so is $R^{(j)}$. It is easy to verify that $\{\bar{D}_k\}$ in \eqref{eqn:Dkmeandef} and $\{H_k\}$ in \eqref{eqn:Hkdef} are all of complex-Hessian-type such that
\be
\label{eqn:LDLandLHL}
L \bar{D}_k L = \bar{D}_k, \quad L H_k L = H_k, \quad k = 1,2,\dots,N
\ee
From \eqref{eqn:Rjdef}, we have
\begin{align}
L R^{(j)} L & \stackrel{(a)}{=} \sum_{k=1}^{N} \sum_{\ell=1}^{N} p_{\ell,k} [ (L \bar{D}_k L) (L R^{(j-1)} L) (L \bar{D}_\ell L)  + c_{\mu, \ell, k} (L H_k L) (L R^{(j-1)} L) (L H_\ell L) ] \nn \\
& \stackrel{(b)}{=} \sum_{k=1}^{N} \sum_{\ell=1}^{N} p_{\ell,k} [ \bar{D}_k R^{(j-1)} \bar{D}_\ell + c_{\mu, \ell, k} H_k R^{(j-1)} H_\ell ] \nn \\
& = R^{(j)}
\end{align}
where step (a) is by the fact that $LL = I_{2M}$ and step (b) is by \eqref{eqn:LDLandLHL} and the induction hypothesis. Therefore, the matrix $R^{(j)}$ is also of complex-Hessian-type.

\section{Proof of Corollary \ref{corollary:steadystateMSD}}
\label{app:steadystateMSD}
Following an argument similar to the proof of Theorem \ref{theorem:steadystateMSD}, we can obtain
\be
\label{eqn:individualMSDp1}
\lim_{i\rightarrow\infty} \E \|\ubar{\wt{\bm{w}}}_{k,i} \|^2 = \lim_{i\rightarrow\infty} \E \| \ubar{\wt{\bm{w}}}_{k,i}' \|^2 + O(\nu^{3/2})
\ee
From \eqref{eqn:individualMSDexpression}, \eqref{eqn:individualMSDp1}, and Corollary \ref{corollary:crosscov}, and using the fact that $\gamma_o = \min \{1/2, \gamma_o'\}$, we can express the individual MSD by
\be
\label{eqn:individualMSDZ}
\MSD_k = \frac{1}{2} \Tr(Z) + O(\nu^{1 + \gamma_o})
\ee
where the first term on the RHS is of the order of $\nu$ and dominates the other term. Then, by \eqref{eqn:networkMSDdef}, we immediately get
\be
\label{eqn:networkMSDZ}
\MSD^\network = \frac{1}{N}\sum_{k=1}^{N} \MSD_k = \frac{1}{2} \Tr(Z) + O(\nu^{1 + \gamma_o })
\ee
Let
\begin{align}
\label{eqn:Sdef}
S & \defeq H^\T \kron I_{2M} + I_{2M} \kron H = O(\nu) \\
\label{eqn:Ydef}
Y & \defeq \sum_{\ell,k=1}^{N} p_{\ell,k} (\bar{\mu}_\ell \bar{\mu}_k  +  c_{\mu,\ell,k})(H_\ell^\T  \kron  H_k) = O(\nu^2)
\end{align}
by \eqref{eqn:bigMbound}--\eqref{eqn:bigCMbound2}, where $H$ is given by \eqref{eqn:Hdef}. It is worth noting that $S$ is invertible when condition \eqref{eqn:boundstepsize4thorder} holds and $Y$ is always invertible by Assumption \ref{I-asm:boundedHessian} in Part I \cite{Zhao13TSPasync1}. By using \eqref{eqn:qXpexpression}, \eqref{eqn:qXpexpression2}, \eqref{eqn:Sdef}, and \eqref{eqn:Ydef}, we get
\be
\label{eqn:IminusFapproxHHO}
I_{4M^2} - F = S + Y
\ee
Using the matrix inversion lemma \cite{Laub05}, we get from \eqref{eqn:IminusFapproxHHO} that
\be
\label{eqn:IminusFinvapprox}
(I_{4M^2} - F)^{-1} = S^{-1} - S^{-1}(Y^{-1} + S^{-1})^{-1}S^{-1}
\ee
By \eqref{eqn:Sdef} and \eqref{eqn:Ydef}, we know that $\| S^{-1} \| = O(\nu^{-1})$ and $\| Y^{-1} \| = O(\nu^{-2})$. Then,
\be
\label{eqn:IminusFinvapprox1}
(I_{4M^2} - F)^{-1} = S^{-1} + O(1)
\ee
By \eqref{eqn:Rdef}, we have $\| R \| = O(\nu^2)$. Using \eqref{eqn:Zdef} and \eqref{eqn:IminusFinvapprox1}, we get
\begin{align}
\label{eqn:traceofZ}
\Tr(Z) & = [\vecm(Z)]^* \vecm(I_{2M}) \nn \\
& = [\vecm(R)]^* (I_{4M^2} - F)^{-1} \vecm(I_{2M}) \nn \\
& = [\vecm(R)]^* S^{-1} \vecm(I_{2M}) + O(\nu^2)
\end{align}
where the first term on the RHS is of the order of $\nu$ and, therefore, is the dominant term. To further simplify \eqref{eqn:traceofZ}, we consider the Lyapunov equation with respect to the unknown matrix $X \in \mbbC^{2M\times 2M}$: $XH + HX = I_{2M}$, where $H$ is given by \eqref{eqn:Hdef}. By applying the $\vecm(\cdot)$ operation to both sides, the Lyapunov equation is equivalent to the linear equation: $S \cdot \vecm(X) = \vecm(I_{2M})$, where $\vecm(X) \in \mbbC^{4M^2\times1}$. Since $S$ is invertible, the Lyapunov equation has a unique solution, which is given by $X = \frac{1}{2} H^{-1}$, or $\vecm(X) = S^{-1} \vecm(I_{2M}) = \frac{1}{2} \vecm(H^{-1})$. It then follows from \eqref{eqn:traceofZ} that
\be
\label{eqn:orderofTrZ}
\Tr(Z) = \frac{1}{2}\Tr(H^{-1}R) + O(\nu^2)
\ee
where the first term on the RHS is of the order of $\nu$ and dominates the other term. Substituting \eqref{eqn:orderofTrZ} into \eqref{eqn:individualMSDZ} and \eqref{eqn:networkMSDZ} completes the proof.

\end{document}

%% file: MyDefinitions.tex
\usepackage{hyperref}
\usepackage{xr}
\usepackage[table]{xcolor}
\usepackage{array}
\usepackage{bbding}
\usepackage{cite}
\usepackage{graphicx}
\usepackage{amssymb}
\usepackage{amsmath}
\usepackage{amsfonts}
\usepackage{algorithmic}
\usepackage{algorithm}
\usepackage{mdwmath}
\usepackage{mdwtab}
\usepackage{eqparbox}
\usepackage{fixltx2e}
\usepackage{stfloats}
\usepackage{dsfont}
\usepackage{mathrsfs}
\usepackage{bbm}
\usepackage{bm}
\usepackage{threeparttable}
\usepackage{multirow}
\usepackage[caption=false,font=footnotesize]{subfig}
\usepackage{centernot}
\usepackage{accents}
\usepackage{cleveref}
\usepackage{lscape}

\newtheorem{theorem}{Theorem}
\newtheorem{assumption}{Assumption}
\newtheorem{lemma}{Lemma}

\newtheorem{corollary}{Corollary}

\newtheorem{definition}{Definition}

\newenvironment{proof}[1][Proof]{\begin{trivlist}
\item[\hskip \labelsep {\bfseries #1}]}{\end{trivlist}}

\newcommand{\qed}{\nobreak \ifvmode \Relax \else
      \ifdim\lastskip<1.5em \hskip-\lastskip
      \hskip1.5em plus0em minus0.5em \fi \nobreak
      \vrule height0.65em width0.65em depth0em\fi}

\def\defeq{\triangleq}

\DeclareMathOperator*{\minimize}{\mathrm{minimize}}

\def\ds{\mathds}
\def\wt{\widetilde}
\def\wh{\widehat}

\def\col{{\mathrm{col}}}

\def\diag{{\mathrm{diag}}}

\def\Tr{{\mathrm{Tr}}}

\def\kron{\otimes}
\def\bkron{\otimes_{b}}
\def\vecm{{\mathrm{vec}}}
\def\unvecm{{\mathrm{unvec}}}
\def\bvec{{\mathrm{bvec}}}
\def\unbvec{{\mathrm{unbvec}}}
\def\MSD{{\textrm{MSD}}}

\def\glob{\textrm{glob}}

\def\network{{\textrm{net}}}

\def\E{\mathbb{E}}
\def\F{\mathbb{F}}

\def\nn{{\nonumber}}

\newcommand{\be}{\begin{equation}}
\newcommand{\ee}{\end{equation}}
\newcommand{\bea}{\begin{eqnarray*}}
\newcommand{\eea}{\end{eqnarray*}}

\makeatletter
\def\Lddots{\mathinner{\mkern1mu\raise17\p@\vbox{\kern17\p@\hbox{.}}\mkern2mu
    \raise8\p@\hbox{.}\mkern2mu\raise\p@\hbox{.}\mkern1mu}}
 \makeatother

\newcommand{\mbbC}{\mathbb{C}}

\newcommand{\mbbF}{\mathbb{F}}

\newcommand{\mbbR}{\mathbb{R}}
\newcommand{\mbbT}{\mathbb{T}}

\newcommand{\Acal}{\mathcal{A}}
\newcommand{\Bcal}{\mathcal{B}}
\newcommand{\Ccal}{\mathcal{C}}
\newcommand{\Dcal}{\mathcal{D}}

\newcommand{\Fcal}{\mathcal{F}}
\newcommand{\Gcal}{\mathcal{G}}
\newcommand{\Hcal}{\mathcal{H}}
\newcommand{\Ical}{\mathcal{I}}
\newcommand{\Jcal}{\mathcal{J}}

\newcommand{\Mcal}{\mathcal{M}}
\newcommand{\Ncal}{\mathcal{N}}
\newcommand{\Pcal}{\mathcal{P}}
\newcommand{\Qcal}{\mathcal{Q}}
\newcommand{\Rcal}{\mathcal{R}}

\newcommand{\Tcal}{\mathcal{T}}

\newcommand{\Xcal}{\mathcal{X}}
\newcommand{\Ycal}{\mathcal{Y}}

\def\Re{\mathfrak{Re}}

\newcommand{\one}{\ds{1}}

\newcommand{\T}{\mathsf{T}} 
\newcommand{\ubar}{\underaccent{\bar}}